\documentclass{article}

\usepackage{arxiv}
\usepackage[utf8]{inputenc} 
\usepackage[T1]{fontenc}    
\usepackage{hyperref}       
\usepackage{url}            
\usepackage{booktabs}       
\usepackage{amsfonts}       
\usepackage{nicefrac}       
\usepackage{microtype}      
\usepackage{lipsum}		
\usepackage{graphicx}
\usepackage{comment}
\usepackage{doi}
\usepackage{amsmath}
\usepackage{amsthm}
\usepackage{thmtools}
\usepackage{cite}
\usepackage[linesnumbered, ruled, vlined]{algorithm2e}

\renewcommand{\undertitle}{}

\newtheorem{theorem}{Theorem}

\title{Lagrange Oscillatory Neural Networks for Constraint Satisfaction and Optimization}
\date{}

\date{
\textsuperscript{1} Microelectronics Dept., LIRMM, University of Montpellier, CNRS, France \\
\vspace{8pt}
\textsuperscript{2} Electrical Engineering Dept., Eindhoven Technical University, The Netherlands \\
\vspace{8pt}
\textsuperscript{*} Corresponding Author: delacour@ucsb.edu \\
}

\author{
    Corentin Delacour\textsuperscript{1,*} 
\And
	Bram Haverkort\textsuperscript{2} 
\And
	Filip Sabo \textsuperscript{2} 
\And
	Nadine Azemard\textsuperscript{1} 
\And
    Aida Todri-Sanial\textsuperscript{1,2}
}

\begin{document}
\maketitle
\begin{abstract}
Physics-inspired computing paradigms are receiving renewed attention to enhance efficiency in compute-intensive tasks such as artificial intelligence and optimization. Similar to Hopfield neural networks, oscillatory neural networks (ONNs) minimize an Ising energy function that embeds the solutions of hard combinatorial optimization problems. Despite their success in solving unconstrained optimization problems, Ising machines still face challenges with \textit{constrained} problems as they can become trapped in infeasible local minima. In this paper, we introduce a Lagrange ONN (LagONN) designed to escape infeasible states based on the theory of Lagrange multipliers. Unlike existing oscillatory Ising machines, LagONN employs additional Lagrange oscillators to guide the system towards feasible states in an augmented energy landscape, settling only when constraints are met. Taking the maximum satisfiability problem with three literals as a use case (Max-3-SAT), we harness LagONN's constraint satisfaction mechanism to find optimal solutions for random SATlib instances with up to 200 variables and 860 clauses, which provides a deterministic alternative to simulated annealing for coupled oscillators. We benchmark LagONN with SAT solvers and further discuss the potential of Lagrange oscillators to address other constraints, such as phase copying, which is useful in oscillatory Ising machines with limited connectivity.
\end{abstract}
\section{Introduction}\label{sec1}
\begin{figure}[t!]
\centering
\includegraphics[width=1\textwidth]{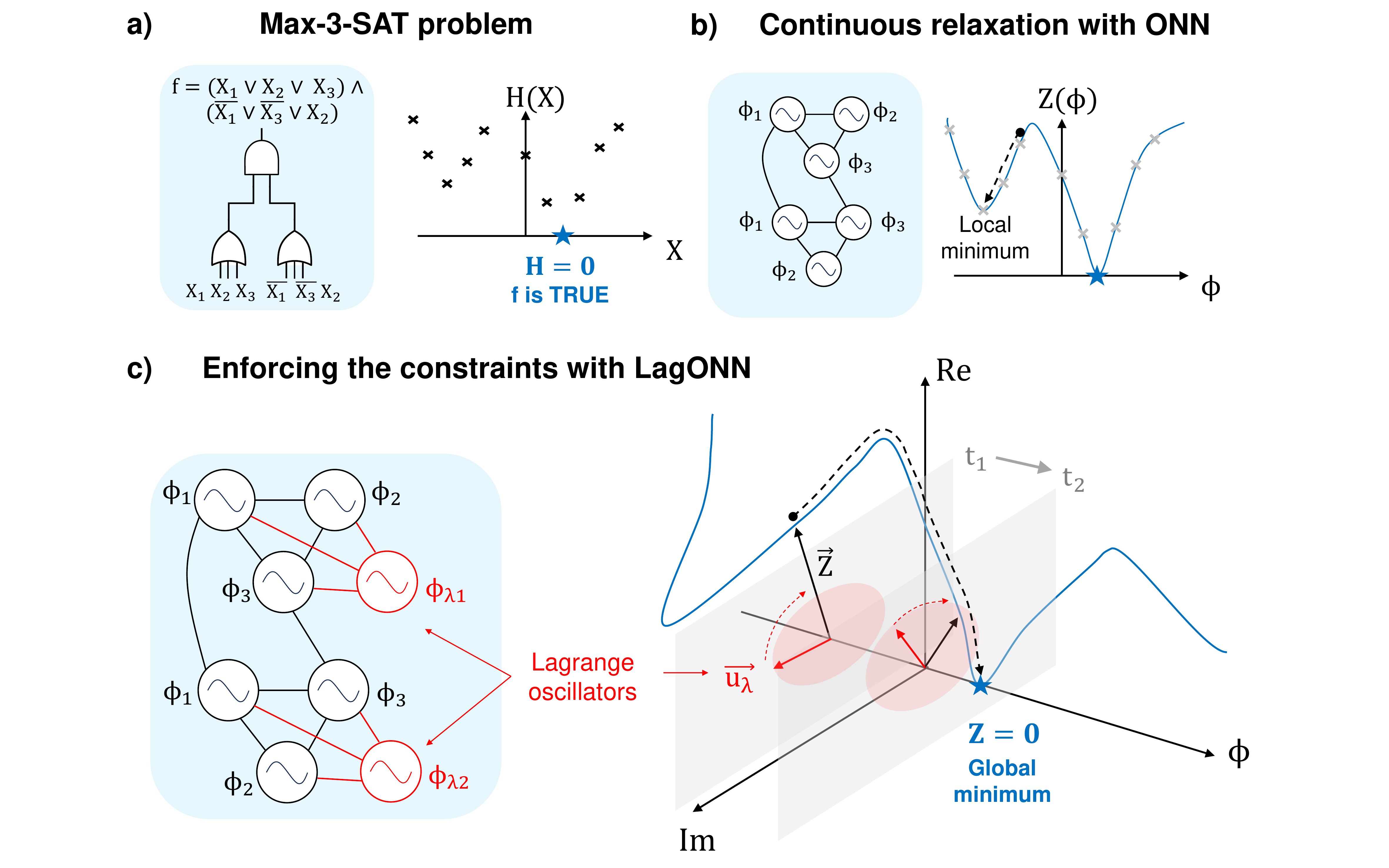}
\caption{a) Any combinatorial optimization problem can be mapped to the 3-SAT problem \cite{Cook_1971}. Its optimization version Max-3-SAT seeks an assignment $X$ satisfying most clauses and can be described as the minimization of an Ising energy function $H(X)$. b) Phase-based oscillatory neural networks (ONN) are analog solvers for combinatorial optimization. Their energy function $E(\phi)$ corresponds to a continuous relaxation of $H(X)$. ONNs can be trapped in local energy minima. c) The proposed Lagrange ONN employs additional oscillators to enforce constraint satisfaction corresponding to 3-SAT clauses. Conceptually, these forces correspond to vectors $\Vec{u_\lambda}$ that "push" an energy vector $\Vec{Z}$ along new directions to escape local minima and reach an optimal solution. }
\label{lagoon_intro}
\end{figure}

Physics-inspired approaches, such as Ising machines, are being actively explored as potential efficient hardware solutions for data-centric applications in artificial intelligence and optimization \cite{Mohseni_2022}. Ising machines are closely related to the seminal works from Hopfield \cite{Hopfiel_1982} and Tank \cite{Tank_1986} in the 1980s, who introduced analog neural networks capable of \textit{naturally} solving NP-hard combinatorial optimization problems, such as the traveling salesman problem (TSP). 
Given a list of cities and their coordinates, TSP seeks a minimum-distance tour that visits every city exactly once and returns to the starting point. Hopfield and Tank's original approach maps cities and their positions in the tour to analog neurons, and intercity distances to synaptic connections, such that the problem instance is "physically wired". The neurons then evolve in parallel to minimize an Ising-like energy that encodes the tour length.
 
Despite its elegance, the system state can become trapped in local energy minima corresponding to infeasible tours, e.g., when two cities are visited simultaneously \cite{Hopfield_1984}. Constraints can theoretically be enforced with sufficiently large penalty coefficients \cite{Lucas_2014}, but this often slows convergence, as the dominance of constraint terms in the energy function hinders the search for optimal solutions. While software implementations can enforce hard constraints by generating only feasible samples, hardware Ising machines typically implement \textit{soft} constraints embedded in the energy function. 
Consequently, current approaches to solving constrained optimization problems with Ising machines rely on tuning the penalty strength to balance feasibility and solution accuracy \cite{Venturelli_2015,parizy_2021,cellini_2024}.

With this energy function, what strategies, other than using impractically large penalty parameters, can we leverage for constraint satisfaction? Due to their continuous energy landscape, analog relaxations may provide techniques beyond those found in binary systems.
In particular, they fit well with the formalism from Lagrange for constrained optimization on continuous functions \cite{hestenes_1969,powell_1978}, which is a common principle to various physics-based solvers \cite{vadlamani_2020}. By relaxing the Ising energy to a continuous energy landscape, it is possible to apply Lagrange's theory and prevent the system state from settling to infeasible solutions, and rather force it to escape infeasible energy minima using additional variables  — Lagrange multipliers $\lambda$  \cite{Nagamatu_1996}. Conceptually, the dynamics consist of a competition between a gradient descent along the original variable direction (to minimize a cost function) and an ascent along the $\lambda$ direction (to enforce the constraints). Another interesting property of continuous variables is the parallel system evolution, contrasting with the typical sequential nature of simulated annealing \cite{Kirkpatrick_1983}, based on Gibbs sampling and only allowing sequential moves between connected spins.

In this paper, we revisit the concept of Lagrange multipliers applied to coupled phase oscillators, which we refer to as \textit{Lagrange oscillatory neural network} (LagONN), as illustrated in Fig.\ref{lagoon_intro}. This platform is motivated by recent advances in coupled oscillator systems \cite{Todri_2024}, which have demonstrated promising results in solving hard combinatorial optimization problems in the analog domain, including the Maximum cut problem \cite{Parihar_2017,Wang_2021,Graber_2022,Graber_2023,Delacour_2023}, Satisfiability (SAT) \cite{Bashar_2023_high_order}, and the Maximum independent set problem \cite{Mallick_2020}. Dense implementations of coupled oscillators using CMOS technology are already available \cite{Ahmed_2021,Moy_2022,Graber_2023,Graber_2024}, while emerging platforms such as spintronic devices \cite{Grollier_2020,Grimaldi_2023} and transition- or bistable-based devices \cite{Dutta_2021,Kim_2023,Maher_2024,Yun_2024} promise further scaling.

The primary objective of this work is to demonstrate constraint satisfaction with coupled phase oscillators without relying on quadratic energy penalties \cite{Lucas_2014}, thereby avoiding the trade-off between feasibility and solution accuracy common in constrained optimization problems.
As an initial demonstration, we select a pure satisfaction problem: the Max-3-SAT problem, whose hardware acceleration has been extensively explored in literature \cite{Sharma_2023,Su_2023, Hizzani_2024,Pedretti_2025,Dikopoulos_2025,Salim_2025}. The objective is to find a Boolean assignment of variable $x$ that maximizes the number of TRUE clauses composed of three literals $C_m=l_1^m\bigvee l_2^m \bigvee l_3^m$ in the formula:
\begin{equation}
    f_B=C_1\bigwedge C_2 \bigwedge...\bigwedge C_{M-1}\bigwedge C_{M}
    \label{boolean_formula_eq}
\end{equation}
With $l_j^m\in \{x_1,...,x_{N},\overline{x_1},...,\overline{x_N}\}$. Assessing the satisfiability of $f_B$ is NP-complete \cite{Cook_1971}, thus any problem in NP can be reduced to 3-SAT in theory. 

For Ising machines, the Max-3-SAT problem can be approached in at least two ways. The first is an unconstrained optimization problem where the goal is to find the ground state of a static Ising energy combining all the clause terms (Fig.\ref{lagoon_intro}b). The second, which is the focus of this work, is to maximize constraint satisfaction via Lagrange multipliers that dynamically enforce the clauses \cite{Nagamatu_1996}, as shown in Fig.\ref{lagoon_intro}c for LagONN.

As the Lagrange formalism provides a deterministic mechanism to escape \textit{infeasible} states, this translates into a ground state search for pure satisfaction problems like Max-3-SAT for which other classical solvers can get stuck at local minima (Fig.\ref{lagoon_intro}c). Consequently, LagONN constitutes a deterministic alternative to stochastic algorithms like simulated annealing \cite{Kirkpatrick_1983} for ground state search. This is an interesting feature for oscillatory-based Ising machines, as harnessing noise and annealing can be challenging in practice. 

The paper is organized as follows. After discussing prior work on physics-inspired solvers, we introduce a Lagrange oscillatory neural network (LagONN) that maps a 3-SAT formula to interconnected oscillatory subcircuits. The circuit dynamics seek an optimal saddle point in the energy landscape that satisfies all constraints.
Next, we benchmark LagONN against SASAT \cite{johnson_1996}, a simulated annealing algorithm, and the SAT solvers WalkSAT \cite{selman_1994} and GNSAT-N \cite{Pedretti_2025} on Max-3-SAT instances from the SATlib library up to 200 variables and 860 clauses \cite{SATlib}. Finally, we discuss the broader utility of Lagrange oscillators by exploring their potential application to other constrained optimization problems, including copying phases in hardware with limited connectivity.

\section{Background} \label{background}
\subsection{Physics-inspired solvers for combinatorial optimization}
At the heart of many physics-inspired solvers for combinatorial optimization lies the idea of mapping the problem's cost function to an energy, often set as the Ising Hamiltonian and expressed as:
\begin{equation}
    H=-\sum_{i<j}J_{ij}S_iS_j-\sum_i h_i S_i
    \label{Ising_ham_eq}
\end{equation}
where $S_i=\pm1$ are the spins, $J_{ij}$ the interaction coefficients between spins, and $h_i$ is the external field applied to spin $S_i$. Finding the ground states of $H$ is in general NP-hard \cite{Lucas_2014}. Updating a spin value $S_i$ causes an energy change $\Delta H_i$ that can drive the search in the energy landscape. While a deterministic search based only on energy minimization will likely get stuck at local minima, most state-of-the-art approaches allow some energy increase with nonzero probability to escape local minima. The standard choice for the system probability distribution is the Boltzmann law $\pi=\exp(-H/T)/Z$ that assigns high probabilities to low-energy states, where $T$ is a temperature parameter and $Z$ is the partition function. Using Markov chain Monte Carlo algorithms such as Gibbs sampling, one can then sample from the Boltzmann distribution by updating each spin sequentially, based on its local input $I_i=\sum_jJ_{ij}S_j+h_i$ \cite{Camsari_2017}. 

Since the goal of combinatorial optimization is to find low-energy states, the simulated annealing algorithm \cite{Kirkpatrick_1983} is a natural choice for hardware implementation on noisy devices such as memristor arrays \cite{Cai_2020} or coupled probabilistic bits (p-bits), which inherently generate probabilistic spin updates \cite{chowdhury_2023}. Implementing Gibbs sampling with synchronous systems typically imposes a sequential spin update, although this is not a limitation for asynchronous systems \cite{Singh_2024} or sparse hardware with high parallelism \cite{Aadit_2022}. In this paper, the system state components evolve in parallel owing to dynamics determined by differential equations. Our model is deterministic, although our numerical integration scheme introduces some errors, as discussed in the Appendix \ref{appendix_solver}.

At first glance, a deterministic search may seem unsuitable for exponentially large spaces as the system state follows a determined trajectory depending on initialization, which may take exponential time before reaching an optimal point. Yet this limitation also applies to simulated annealing, where finite cooling schedules can trap variables in suboptimal states \cite{Aarts_1989}. Other adiabatic approaches propose to slowly anneal the coupling amplitudes $|J_{ij}|$ in real-time, effectively shaping the energy landscape \cite{fahimi_2021,jiang_2023}, or to induce bifurcation phenomena during the adiabatic evolution of non-linear Hamiltonian systems, as in the simulated bifurcation algorithm \cite{Goto_2019,ansari_2024}. In our approach, the weight amplitude $|J_{ij}|$ is rather fixed, but its \textit{phase} $\theta_{ij}$ varies in real-time, which can also be seen as an adaptive energy landscape.

Finally, many deterministic solvers rely on tailored differential equations \cite{Ercsey-ravasz_2011,Molnar_2013,Traversa_2017,Goto_2019}, where attractors in the phase space correspond to solutions, or deliberately introduce chaos to mimic Boltzmann sampling \cite{Suzuki_2013,Lee_2025}. A key consideration for building a corresponding physical machine is whether the system has bounded or unbounded variables. In \cite{Ercsey-ravasz_2011} and \cite{Molnar_2018}, the authors propose an analog model that provides an exponential speed-up, but at the cost of unbounded variables and thus potentially exponential power. In practice, variable growth is capped by the finite power supply \cite{chang_2022}, which can lead to exponential computation times on hard instances \cite{Molnar_2013}.
In this work, we focus on dynamics realizable with coupled phase oscillators, ensuring bounded circuit power since variables are phases and the synaptic amplitudes $|J_{ij}|$ remain fixed. Consequently, we do not anticipate exponential speed-up, as confirmed by our time-to-solution study on Max-3-SAT.

\subsection{Computing with analog oscillatory neural networks}
This work focuses on oscillatory neural networks operating in the phase domain: assuming identical oscillator frequencies, information is encoded in the phases $\phi_i$ relative to a reference oscillator. Under this assumption, and with symmetric synaptic weights $J_{ij}=J_{ji}$, coupled oscillators minimize an Ising-like energy and are similar to analog Hopfield neural networks \cite{Izhikevich_2006,Jackson_2018,Abernot_2021}. 
With sinusoidal interactions, the network energy takes the form of a two-dimensional XY Ising Hamiltonian \cite{Wang_2019} :
\begin{equation}
    E=-\sum_{i<j}J_{ij}\cos(\phi_i-\phi_j)-\sum_i h_i \cos(\phi_i)
    \label{ONN_energy_eq}
\end{equation}
which is an analog relaxation of the Ising model (Eq. \ref{Ising_ham_eq}), since $E=H$ when the phases are binary (multiple of $\pi$).
The connection between coupled oscillators and the Ising model can be illustrated with two coupled oscillators in the absence of external fields. For a coupling $J<0$, the energy $E=-J\cos(\phi_1-\phi_2)$ is minimized when the phase difference is $\pi$ (corresponding to opposite Ising spins), whereas for $J>0$, the minimum occurs when the phases align (corresponding to identical spins). Oscillatory neural networks typically reach these energy minima via gradient descent, governed by:
\begin{equation}
    \dot{\phi}=-\nabla_\phi E
\end{equation}

While the continuous nature of phases in $E$ has been exploited in efficient Max-Cut solvers \cite{Burer_2001,Erementchouk_2022}, most studies of coupled oscillators focus on the binary Ising model, often introducing harmonic signals to binarize phases, as proposed in \cite{Wang_2021}. In this work, we do not employ an additional binarization mechanism. As we will show, the fixed points naturally tend towards binary phases as the number of constraints (3-SAT clauses) increases.

In \cite{Bashar_2023_high_order}, the authors extend the energy $E$ to hypergraphs, incorporating high-order interactions between variables, whereas the classical Ising model (Eq. \ref{Ising_ham_eq}) includes only pairwise (quadratic) spin interactions. This extension allows certain NP-hard problems, such as Satisfiability, to be mapped directly to hardware without requiring quadratization, which adds auxiliary oscillators \cite{Bybee_2023}. In this work, we focus on this approach, specifically considering simultaneous interactions among four oscillators $k, l, m, n$ whose energy can be written as:
\begin{equation}
    E_{k l m n}=J_{k l}\cos(\phi_k-\phi_l+\phi_m-\phi_n)=Re\Big(J_{k l}\exp\big[i(\phi_m-\phi_n)\big]\exp\big[i(\phi_k-\phi_l)\big]\Big)
    \label{high_order_eq}
\end{equation}
Examining Eq. \ref{high_order_eq}, the additional third- and fourth-order interactions ($m$ and $n$) can be interpreted as a complex synapse connecting $k$ and $l$ with amplitude $J_{kl}$ and synaptic phase $\theta_{kl}=\phi_m-\phi_n$, which conveys the high order information from $m$ and $n$. This phase difference can also be interpreted as a delay $\tau_{kl}=(\phi_m-\phi_n)/\omega_0$ where $\omega_0$ is the oscillating frequency. With fourth-order interactions, the three variables of a 3-SAT clause plus an additional Lagrange variable could interact simultaneously. We note that implementing such high-order interaction is nontrivial and leave it as future work, though some conceptual ideas are discussed in Section \ref{sec_arch}.

\subsection{Constrained optimization with Lagrange multipliers}

This work leverages the concept of Lagrange multipliers that are at the heart of many approaches to constrained optimization \cite{hestenes_1969,powell_1978,fisher_2004,Wah_1999}. For a cost function $f(x)$ with constraints $g(x)$, these methods typically define a Lagrange function:
\begin{equation}
    L(x,\lambda)=f(x)+\lambda^Tg(x)
    \label{lagrange_function}
\end{equation}
where the constraints are satisfied when $g(x)=0$, and $\lambda$ is a new variable called the \textit{Lagrange multiplier}. Similar to the penalty method \cite{hestenes_1969}, the Lagrange approach weights the constraints and adds them to the cost function. The minimum of $L$ with respect to $x$ then provides a lower bound for the optimal constrained solution $f(x^*)$, since $\min_x L\leq L(x^*,\lambda)=f(x^*)$.
The purpose of the Lagrange multiplier $\lambda$ is to close this gap by finding an optimal $\lambda^*$ such that $\min_x L(x,\lambda^*)=f(x^*)$, or at least to reduce the gap as much as possible, yielding the tightest lower bound on the optimal value \cite{fisher_2004}. Because $L$ is concave in $\lambda$ \cite{boyd_2023}, the optimal $\lambda^*$ can be obtained by maximizing $\min_xL$:
\begin{equation}
    L^*=\max_\lambda \min_x L(x,\lambda)
    \label{max_min_eq}
\end{equation}

For continuous functions $f$ and $g$, this optimization problem can be interpreted as a Lagrange neural network, as defined in \cite{Zhang_1992}, where neurons perform gradient descent along $x$ to minimize $L$, while Lagrange neurons perform gradient ascent in the $\lambda$-direction to reduce the gap (Eq. \ref{max_min_eq}) and reach an optimal solution $f(x^*)$ that satisfies the constraints. The corresponding dynamics are:
\begin{equation}
\begin{cases}
        \tau \dot{x}=-\nabla_x L \\
        \tau_\lambda\dot{\lambda}=+\nabla_\lambda L=g(x)
\end{cases}
        \label{LNN_dynamics}
\end{equation}
where $\tau$ and $\tau_\lambda$ control the speeds of minimization and maximization, respectively. Reaching the global optimum $f(x^*)$ requires an ideal minimizer, and in non-convex settings, gradient descent can be trapped in local minima, yielding suboptimal solutions. Nevertheless, it is important to note that the search for a constrained solution \textit{cannot} stop until the constraints are satisfied, i.e., $g(x)=0$.
Hence, for satisfaction problems such as Max-3-SAT, which we explore in this paper, the search continues until all clauses are TRUE. Another geometric interpretation of the dynamics in Eq. \ref{LNN_dynamics} is that they seek a saddle point of $L$, a minimum in the $x$-direction and a maximum in the $\lambda$-direction. In general, the energy landscape is non-convex and $L^*$ is not necessarily a saddle point \cite{boyd_2023}. For LagONN, however, we show that for satisfiable instances, $L^*$ corresponds to a saddle point that satisfies all constraints and can be reached via the competitive dynamics defined in Eq. \ref{LNN_dynamics}.

\section{Results}
\subsection{Mapping 3-SAT to oscillatory neural networks}

In this work, we build on the approach of Nagamatu et al. \cite{Nagamatu_1996} to solve the Max-3-SAT problem (Eq. \ref{boolean_formula_eq}) using a Lagrange neural network with phase-based oscillatory neurons. Before introducing Lagrange multipliers, we now describe how to map a 3-SAT clause to coupled oscillators. First, we express each clause $C_i$ as an Ising energy $H_i$ that equals 0 if and only if $C_i$ is satisfied. For 3-SAT, there are four possible clause types, which we map to Ising energies as follows:
\begin{align}
&C_1=X\vee Y\vee Z \longrightarrow H_1=1+S_XS_Y+S_XS_Z+S_YS_Z-(S_X+S_Y+S_Z)-S_XS_YS_Z \\ 
&C_2=\overline{X}\vee Y\vee Z \longrightarrow H_2=1-S_XS_Y-S_XS_Z+S_YS_Z-(-S_X+S_Y+S_Z)+S_XS_YS_Z \nonumber\\
&C_3=\overline{X}\vee \overline{Y}\vee Z \longrightarrow H_3=1+S_XS_Y-S_XS_Z-S_YS_Z-(-S_X-S_Y+S_Z)-S_XS_YS_Z\nonumber \\
&C_4=\overline{X}\vee \overline{Y}\vee \overline{Z} \longrightarrow H_4=1+S_XS_Y+S_XS_Z+S_YS_Z+(S_X+S_Y+S_Z)+S_XS_YS_Z \nonumber
\end{align}
where spins $S_j=\pm1$ are the binary Ising variables with $S_j=+1$ corresponding to TRUE. Writing the truth table for each clause, we find that $C_i$ is TRUE when $H_i=0$ and $C_i$ is FALSE when $H_i=8$.
The next step is to map these Ising energies onto phase-based oscillatory neural networks. We propose relaxing the binary Ising Hamiltonians to complex variables defined as:
\begin{figure}[t!]
\centering
\includegraphics[width=1\textwidth]{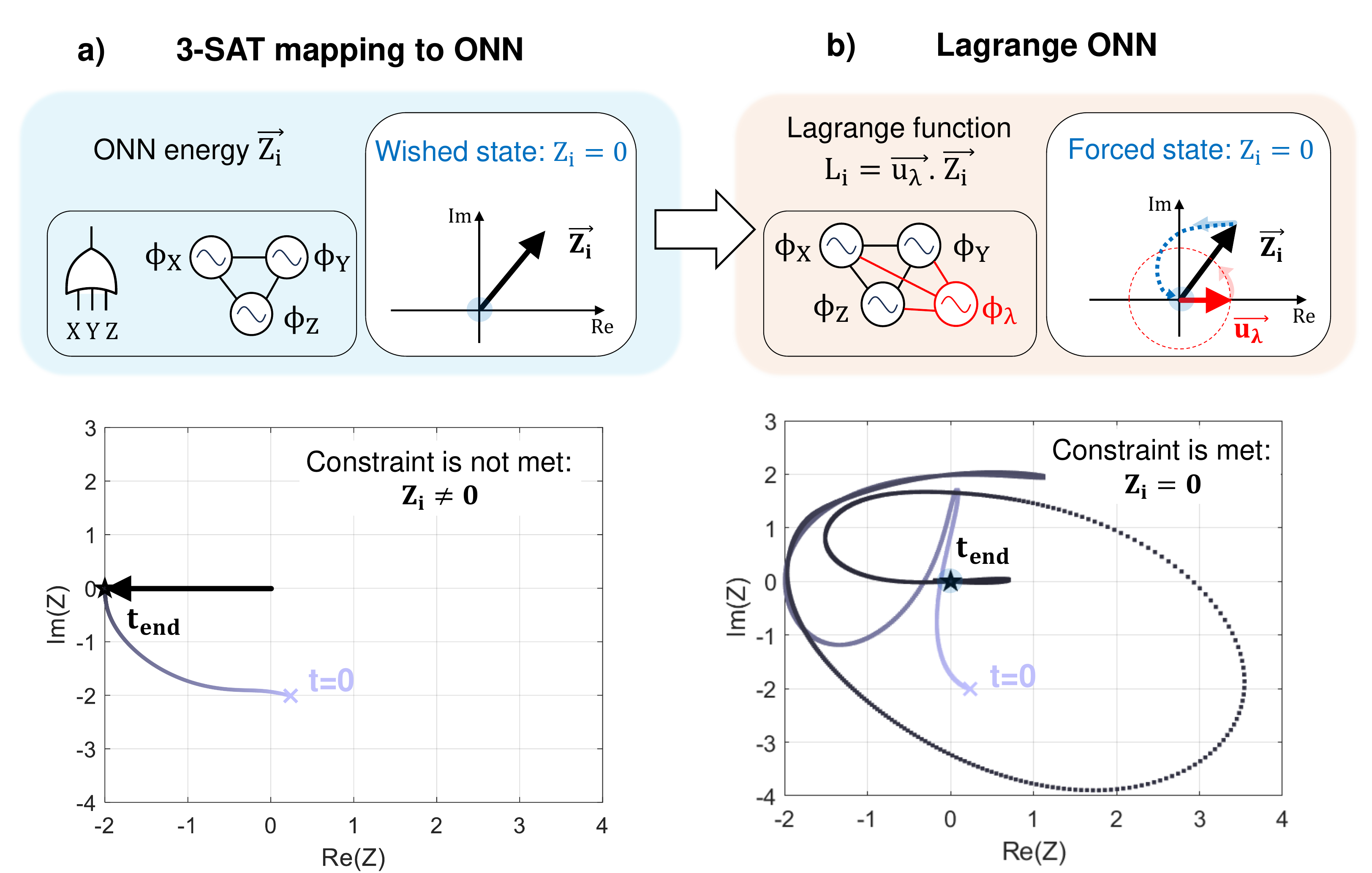}
\caption[LagONN Principle]{a) 3-SAT clause mapping to ONN. The Ising energy is relaxed to a complex quantity $Z_i$, a function of ONN phases $\phi_X,\phi_Y,\phi_Z$. For binary phases, $Z_i=0$ corresponds to the optimal Ising state and induces $C_i=$ TRUE. However, the standard ONN trajectory settles to an undesired fixed point where $Z_i\neq 0$ (bottom). b) Adding a Lagrange oscillator with phase $\phi_\lambda$ can enforce the constraint $Z_i=0$. This is achieved by defining the Lagrange function $L_i$ and setting competitive dynamics (gradient descent and ascent) seeking a saddle point of $L_i$ where $Z_i=0$. Bottom: resulting Lagrange ONN trajectory for the same phase initialization.}
\label{LagONN_principle}
\end{figure}
\begin{align}
&H_1 \longrightarrow Z_1=1+e^{i(\phi_X-\phi_Y)}+e^{i(\phi_X-\phi_Z)}+e^{i(\phi_Z-\phi_Y)}-(e^{i\phi_X}+e^{i\phi_Y}+e^{i\phi_Z})-e^{i(\phi_X-\phi_Y+\phi_Z)} \label{complex_relaxation} \\ 
&H_2 \longrightarrow Z_2=1-e^{i(\phi_X-\phi_Y)}-e^{i(\phi_X-\phi_Z)}+e^{i(\phi_Z-\phi_Y)}-(-e^{i\phi_X}+e^{i\phi_Y}+e^{i\phi_Z})+e^{i(\phi_X-\phi_Y+\phi_Z)} \nonumber \\
&H_3 \longrightarrow Z_3=1+e^{i(\phi_X-\phi_Y)}-e^{i(\phi_X-\phi_Z)}-e^{i(\phi_Z-\phi_Y)}-(-e^{i\phi_X}-e^{i\phi_Y}+e^{i\phi_Z})-e^{i(\phi_X-\phi_Y+\phi_Z)} \nonumber \\
&H_4 \longrightarrow Z_4=1+e^{i(\phi_X-\phi_Y)}+e^{i(\phi_X-\phi_Z)}+e^{i(\phi_Z-\phi_Y)}+(e^{i\phi_X}+e^{i\phi_Y}+e^{i\phi_Z})+e^{i(\phi_X-\phi_Y+\phi_Z)} \nonumber
\end{align}

We deliberately introduce phase differences, e.g. $\phi_X-\phi_Y$, to define a complex relaxation $Z_i$ of the conventional coupled-oscillators energy function (Eq. \ref{ONN_energy_eq}). By construction, the complex relaxation $Z_i$ equals the Hamiltonian $H_i$ for binary phases $\phi_j=\pi(1-S_j)/2$.
However, non-binary phases can exist such that $Z_i=0$, making the assignment of phases to Ising spins ambiguous. This issue is not caused by the high-order interaction terms but is inherent to the two-dimensional nature of the XY Ising model. A similar situation arises when solving the Max-cut problem with coupled oscillator systems without forcing binarization \cite{Wang_2021,Delacour_2023}, raising the question of how to round phases to spins.

In this work, we do not face this issue, as for more than a few clauses, the phase fixed points tend to be binary due to an overdetermined system of $2M$ equations and $N$ phase variables, where $M>N$ and $M$ is the number of 3-SAT clauses.
Leveraging this property, we design a coupled-oscillator module that minimizes $|Z_i|$, so that when $Z_i=0$, the phases $\phi_X$, $\phi_Y$, and $\phi_Z$ directly yield the Boolean values $S_X$, $S_Y$, and $S_Z$ satisfying clause $C_i$.

\subsection{Enforcing constraints with an oscillatory-based Lagrange multiplier}

Constraining the oscillators to satisfy a clause $C_i$ with energy $Z_i$ can be expressed via the Lagrange function $L_i(\phi,\lambda)$ as:
\begin{equation}
L_i(\phi,\lambda)=\lambda_{R} \operatorname{Re}[Z_i]+\lambda_{I} \operatorname{Im}[Z_i]
\label{ONN_Lagrange}
\end{equation}
where $\lambda_R$ and $\lambda_I$ are the Lagrange multipliers for the constraints, satisfied when $\nabla_\lambda L_i=(\operatorname{Re}[Z_i], \, \operatorname{Im}[Z_i])=0 \iff Z_i=0$. Note that there are only two constraints and no cost function $f$, in contrast to the general Lagrange multiplier formulation (Eq. \ref{lagrange_function}). This defines a pure constraint satisfaction problem with constraint function $g=(\operatorname{Re}[Z_i], \operatorname{Im}[Z_i])=0$.

There are two possible interpretations for the Lagrange multipliers, which give rise to two distinct types of Lagrange oscillatory neural networks:
\begin{enumerate}
    \item $\lambda_R$ and $\lambda_I$ are \textit{synaptic elements}. This raises the question of how to implement synapses $\lambda_R$ and $\lambda_I$ that evolve in real-time while having a limited range in a physical system.
    \item $\lambda_R$ and $\lambda_I$ are \textit{oscillatory variables}, i.e. phase oscillators. In this interpretation, only oscillating neurons and synapses with fixed amplitude are involved, potentially simplifying hardware implementation.
\end{enumerate}
Focusing on the second interpretation, a specific choice for $\lambda$ simplifies the Lagrange function $L_i$ (Eq. \ref{ONN_Lagrange}). We consider a Lagrange \textit{oscillator} with the same frequency as the other neurons, phase $\phi_\lambda$, and unit amplitude (Fig. \ref{LagONN_principle}b). In the 2D plane, its corresponding unit vector $\Vec{u_\lambda}$ has coordinates $(\cos\phi_\lambda, \, \sin \phi_\lambda)$, which we assign to multipliers $(\lambda_R, \, \lambda_I)$. Consequently, the Lagrange function reduces to the dot product between $\Vec{u_\lambda}$ and the vector $\Vec{Z_i}=(\operatorname{Re}[Z_i],\,\operatorname{Im}[Z_i])$ as:
\begin{equation}
    L_i(\phi,\phi_\lambda)=\Vec{u_\lambda} . \Vec{Z_i}=\cos\phi_\lambda \operatorname{Re}[Z_i]+\sin\phi_\lambda \operatorname{Im}[Z_i]
    \label{LagONN_vec_equation}
\end{equation}
This yields a compact form for $L_i$, expressed here for the first type of clause ($i=1$):
\begin{equation}
\begin{aligned}
    L_1(\phi,\phi_\lambda)=&\cos\phi_\lambda\\
    -&\cos(\phi_X-\phi_\lambda)-\cos(\phi_Y-\phi_\lambda)-\cos(\phi_Z-\phi_\lambda)\\
    +&\cos(\phi_X-\phi_Y-\phi_\lambda)+\cos(\phi_X-\phi_Z-\phi_\lambda)+\cos(\phi_Z-\phi_Y-\phi_\lambda)\\
    -&\cos(\phi_X-\phi_Y+\phi_Z-\phi_\lambda)
\end{aligned}
\label{Lyapunov_LagONN}
\end{equation}
We identify the energy function as that of four sinusoidal coupled oscillators $\phi_X, \,\phi_Y, \,\phi_Z, \,\phi_\lambda$ with high-order interactions up to fourth order, as indicated by the last term and previously introduced in Eq. \ref{high_order_eq}. Notably, the synaptic amplitudes appear as fixed binary weights $\pm1$ for the cosine terms, ensuring that the LagONN variables remain bounded, since only phases vary and are bounded by definition.

Before deriving the corresponding circuit, we now analyze the properties of the energy landscape $L_i$. Introducing an oscillatory Lagrange multiplier generates a saddle point in the energy landscape where the constraint $Z_i=0$ is satisfied, as formalized in the next theorem.

\begin{theorem} \label{theorem1}
Let $L_i(\phi,\phi_\lambda)=\Vec{u_\lambda}.\Vec{Z_i}$ then: \

\begin{enumerate}
    \item 
    $L_i\text{ has a at least one saddle point } L_i(\phi^*,\phi_\lambda^*)=0$ such that $L_i(\phi^*,\phi_\lambda) \leq L_i(\phi^*,\phi_\lambda^*) \leq L_i(\phi,\phi_\lambda^*)$.
   \item Such saddle point satisfies the constraint $Z_i(\phi^*)=0$.
\end{enumerate}
\end{theorem}

\begin{figure}[t!]
\centering
\includegraphics[width=\textwidth]{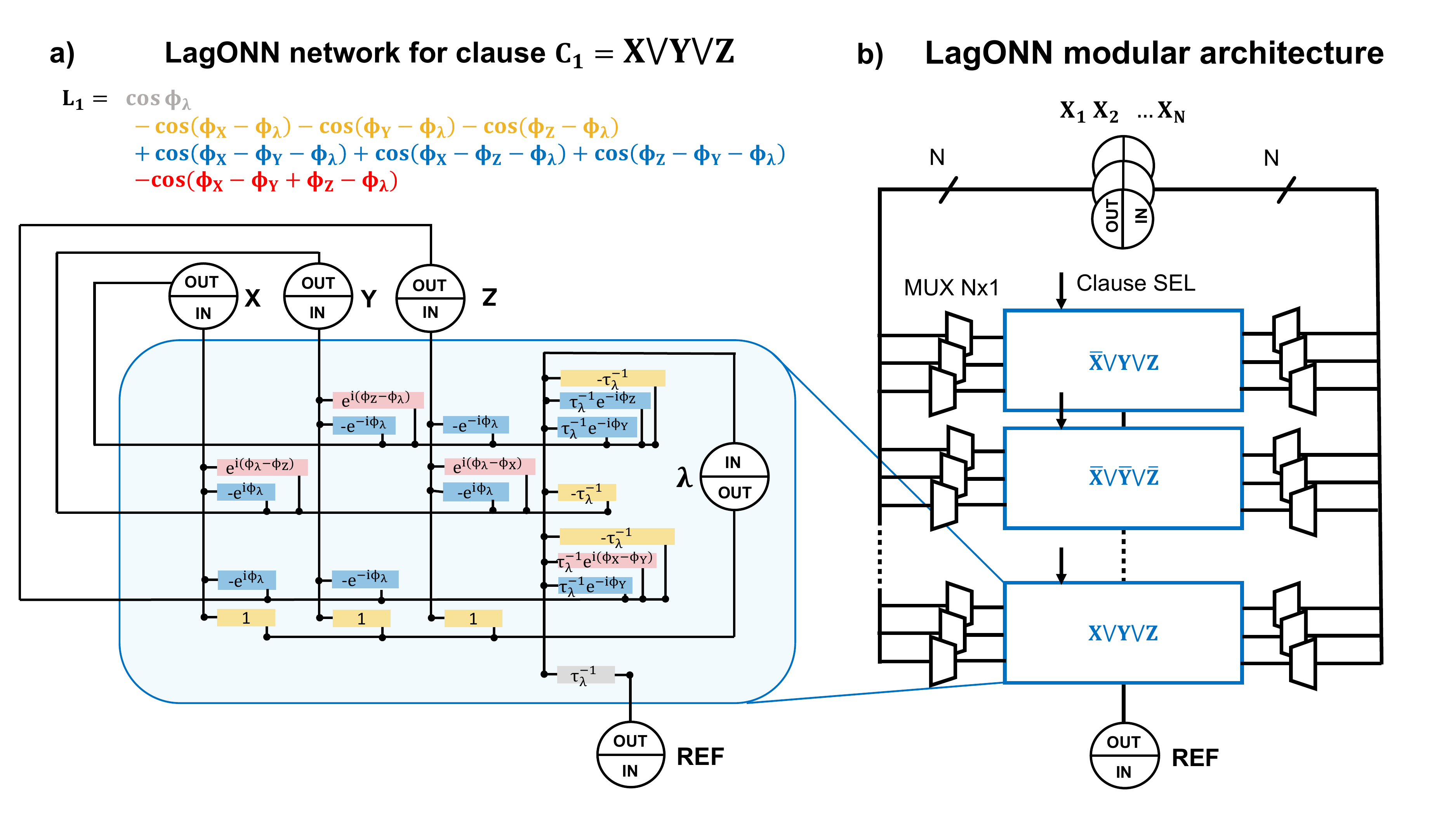}
\caption[LagONN Network]{a) LagONN network for solving clause $C_1$. Rectangles denote complex synapses in the form $J_{ij}\exp{(i\theta_{ij}})$ where $\theta_{ij}$ is implemented with a delay $\theta_{ij}/\omega_0$ in practice. The connections for $\theta_{ij}$ are not shown. b) Modular LagONN architecture to program any 3-SAT formula $f_B$ with $N$ variables and $M$ clauses. Each rectangle corresponds to a clause subcircuit and contains a Lagrange oscillator. The clause selection signal programs the clause subcircuit to one of the four possible clauses. $N\times1$ multiplexers route oscillator input and output signals to each clause according to the desired Boolean formula $f_B$.}
\label{LagONN_network}
\end{figure}

The theorem is proven in Appendix \ref{appendix_dynamics}. Since LagONN's energy landscape $L_i(\phi,\phi_\lambda)$ has at least one saddle point satisfying the constraint, we set the phase dynamics for $\phi$ such that it minimizes $L_i(\phi,\phi_\lambda)$, while the Lagrange oscillator with phase $\phi_\lambda$ maximizes $L_i(\phi,\phi_\lambda)$. This aims to reach a saddle point where $Z_i(\phi^*)=0$, as shown in Fig.\ref{LagONN_principle}b. Accordingly, we propose the following phase dynamics for clause $C_i$:
\begin{equation}
    \begin{cases}
        \tau \dot{\phi_j}=-\nabla_{\phi_j} L_i(\phi,\phi_\lambda)=-\Vec{u_\lambda}.\partial \Vec{Z_i}/\partial \phi_j \\
        \tau_\lambda\dot{\phi_\lambda}=+\nabla_{\phi_\lambda} L_i(\phi,\phi_\lambda)=\Vec{u_\lambda'}.\Vec{Z_i}
    \end{cases}
    \label{chap3_LagONN_dynamics}
\end{equation}
where $\tau$ and $\tau_\lambda$ are the time constants for the standard and Lagrange oscillators, and $\Vec{u_\lambda'}=(-\sin\phi_\lambda,\cos\phi_\lambda)$. The time constants define the relative speed between gradient descent and ascent, which we take as $\tau=\tau_\lambda$ to achieve faster convergence in simulations (see Appendix \ref{appendix_lagrange_speed}).
The dynamics can be interpreted as a \textit{competition} between $\Vec{Z_i}$ and $\Vec{u_\lambda}$, with the desired outcome being $Z_i=0$ ($L_i$'s optimal saddle point) as a compromise between the two competing dynamics.

Fig.\ref{LagONN_principle} shows examples of oscillatory neural network dynamics without and with a Lagrange oscillator. Without the Lagrange oscillator, the phases evolve to minimize $\operatorname{Re}[Z_i]$ and settle to an undesired fixed point where $Z_i=-2$ (Fig.\ref{LagONN_principle}a). However, the desired state is $Z_i=0$ (Eq. \ref{complex_relaxation}), which would satisfy the clause $C_i$ for binary phases. Introducing a Lagrange oscillator as defined in Eq. \ref{Lyapunov_LagONN}, with competitive dynamics, constrains the phases to reach a saddle point satisfying $Z_i=0$. Fig.\ref{LagONN_principle}b shows an example of LagONN for the same initialization, highlighting the complex dynamics arising from this competition. Ultimately, the system settles to an optimal saddle point where $Z_i=0$. Note that the Lagrange function (Eq. \ref{LagONN_vec_equation}) is not a Lyapunov function for the system, as it can increase over time (see Appendix \ref{appendix_dynamics}).

\subsection{Modular LagONN architecture for 3-SAT formula}\label{sec_arch}
\begin{table}[b!]
\caption{Comparison between two possible LagONN architectures.}
\centering
\begin{tabular}{|c|c|c|}
\hline
LagONN architecture                & Fully-connected & Modular               \\ \hline
Oscillators                        & $N+M$         & $N+M$                  \\
Size of synaptic array             & $(N+M)^2$       & $4^2$/module (clause) \\
Distance of high-order interaction & $N+M$ (global)  & $4$ (local)           \\
$N\times 1$ multiplexers             & 0               & $6M$                 \\ \hline
\end{tabular}
\label{chap3_tab:LagONN_arch}.
\end{table}
We now propose a circuit implementation for a LagONN clause $C_i$ with energy $L_i$ and the competitive dynamics previously introduced (Eq. \ref{chap3_LagONN_dynamics}). A network implementing a single clause as in Eq. \ref{chap3_LagONN_dynamics} for $C_1=X\bigvee Y\bigvee Z$ is shown in Fig.\ref{LagONN_network}a.
The network includes a reference oscillator to measure phases and apply an external field to the Lagrange oscillator, so that in practice $\phi_j$ corresponds to $\phi_j-\phi_{REF}$. The main source of LagONN complexity is the synaptic array, which consists of delayed and weighted signals. In Fig.\ref{LagONN_network}a, a synapse $S_{ij}$ connecting oscillators $i$ and $j$ with weight $W_{ij}$ and phase $\theta_{ij}$ is represented as $S_{ij}=W_{ij}e^{i\theta_{ij}}$. This supposes having a mechanism to delay the synaptic input by $\theta_{ij}/\omega_0$ in real-time, which is the consequence of the fourth-order interaction terms in the LagONN function (Eq. \ref{Lyapunov_LagONN}).

 We believe this scheme is compatible with mixed-signal ONNs \cite{Moy_2022,Delacour_2023}, which could use digital synchronization circuits such as latches or counters to propagate the phase information from the third and fourth oscillators. Another approach would be to modulate the synaptic current amplitude between two oscillators in real-time with the phases of two others. Such circuitry would likely employ analog amplifiers, compatible with a variety of analog oscillators, including spintronic-based \cite{Grollier_2020} or relaxation oscillators \cite{Dutta_2021,Corti_2021,Delacour_2023}.
 
We now introduce the circuit for a larger Boolean formula $f_B$ with $N$ Boolean variables and $M$ clauses $C_m=l_1^m\bigvee l_2^m \bigvee l_3^m$ defined as:
\begin{equation}
    f_B=C_1\bigwedge C_2 \bigwedge...\bigwedge C_{M-1}\bigwedge C_{M}
\end{equation}
with $l_j^m\in \{x_1,...,x_{N},\overline{x_1},...,\overline{x_N}\}$.
The 3-SAT instance can be mapped to LagONN modules where a module $m$ corresponds to a clause $C_m$, and each literal $l_j^m$ corresponds to an input port of that module (Fig.\ref{LagONN_network}b). Each module implements the network shown in Fig.\ref{LagONN_network}a and is highlighted by the blue box. Implementing the AND operation "$\bigwedge$" between two clauses is unnecessary, as each Lagrange oscillator evolves to satisfy its respective clause. Consequently, the entire network maximizes the number of TRUE clauses in $f_B$.

LagONN modules are connected as follows. If two clauses $m$ and $n$ share a literal at positions $k$ and $l$: either identical $l_k^m=l_l^n$ or negated $l_k^m=\overline{l_l^n}$ corresponding to variable $\phi_{x_j}$, $j\in\{1,...,N\}$, we connect input ports $k$ and $l$ so that the synapses from both clauses influence $\phi_{x_j}$.
Repeating this procedure for every pair of clauses ensures that each variable $\phi_{x_j}$ is influenced by all corresponding clauses, yielding the total LagONN function:
\begin{equation}
    L_T(\phi_x,\phi_\lambda)=\sum_{m=1}^M\Vec{u_\lambda ^m}.\Vec{Z_m}(\phi_x)
    \label{lagoon_Lagrange_function}
\end{equation}
where $\phi_x=(\phi_{x_1},\phi_{x_2},...,\phi_{x_{N-1}},\phi_{x_N})^T$ is the vector of phases corresponding to the Boolean variables 
and $\phi_\lambda=(\phi_{\lambda_1},\phi_{\lambda_2},...,\phi_{\lambda_{M-1}},\phi_{\lambda_M})^T$ contains all the Lagrange variables.

The proposed modular architecture offers two main advantages. First, it avoids the need for large synaptic arrays, as synapses are confined within each module. Second, it preserves high-order synaptic interactions locally within the modules.
In contrast, a programmable fully-connected design with $N+M$ oscillators would require a $(N+M)^2$ synaptic array to support any $f_B$ instance. Propagating high-order interactions across such a large array is particularly challenging, as it is not straightforward with a standard two-dimensional grid layout.
However, the modular approach introduces a different type of two-dimensional array: since each module input/output can be connected to $N$ different input/output lines, a programmable architecture would require $6M$ $N\times 1$ multiplexers. Thus, compared to a fully connected design, there is a trade-off between the overhead from multiplexers and the complexity of propagating high-order interactions. The comparison between the modular and fully connected LagONN architectures scaling is summarized in Table \ref{chap3_tab:LagONN_arch}.

\begin{figure}[t!]
\centering
\includegraphics[width=0.8\textwidth]{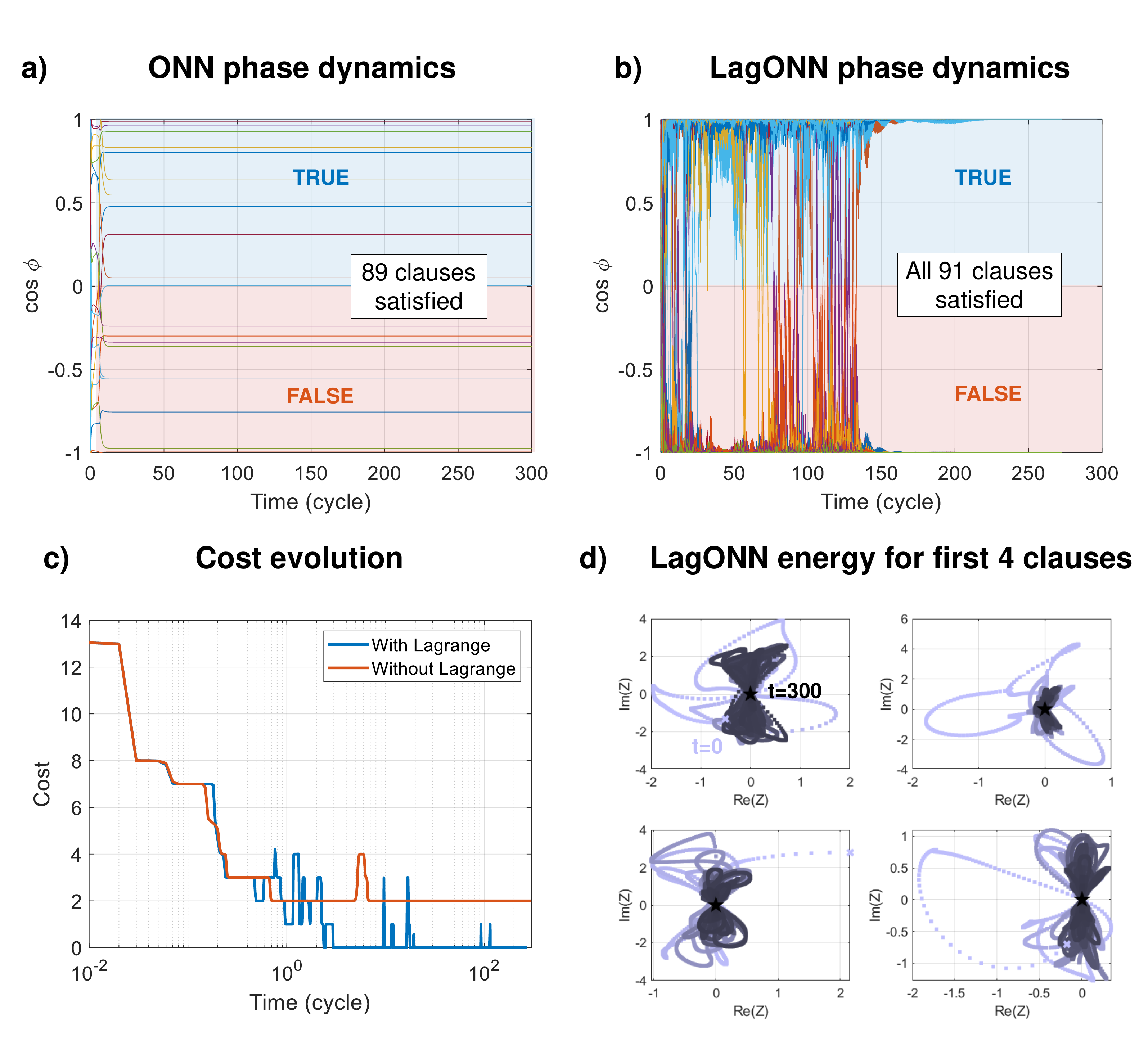}
\caption{Phase dynamics for the satisfiable SATlib instance 'cnf-20-01' with $N=20$ variables and $M=91$ clauses. a) Standard ONN dynamics quickly converge toward a sub-optimal solution where 89/91 of clauses are satisfied. b) With the same initialization, the Lagrange version takes more time to reach a fixed point. Reading out the phases gives an optimal Boolean assignment where all clauses are satisfied. c) Cost function comparison between the two approaches. LagONN finds an assignment of optimal phase around the same time the standard ONN settles ($\approx$ 10 oscillation cycles). By measuring the cost in real time, we can stop the run when a target cost is reached without waiting for convergence. d) Dynamics for four LagONN energy terms $Z_m$ corresponding to the first four clauses. While the dynamics almost seem chaotic, they evolve to reach a target saddle point where all $Z_m=0$. Ultimately, the dynamics converge toward a fixed point at t=300 oscillation cycles where all $Z_m=0$.}
\label{LagONN_sim}
\end{figure}

\subsection{LagONN competitive dynamics}

We now express the phase dynamics of the entire LagONN circuit shown in Fig. \ref{LagONN_network}b. Assuming that the formula $f_B$ is satisfiable, we extend Theorem \ref{theorem1} to the total Lagrange function $L_T$ with $M$ clauses.
\begin{restatable}{theorem}{originaltheorem} \label{theorem2}
Let $L_T(\phi,\phi_\lambda)=\sum_m^M\Vec{u_\lambda^m}.\Vec{Z_m}$ and the formula $f_B=C_1\bigwedge C_2 \bigwedge...\bigwedge C_{M-1}\bigwedge C_{M}$ is satisfiable, then: \

\begin{enumerate}
    \item 
    $L_T\text{ has a at least one saddle point } L_T(\phi^*,\phi_\lambda^*)=0$ such that $L_T(\phi^*,\phi_\lambda) \leq L_T(\phi^*,\phi_\lambda^*) \leq L_T(\phi,\phi_\lambda^*)$.
   \item Such saddle point satisfies the constraints $Z_m(\phi^*)=0$ for all clauses.
\end{enumerate}
\end{restatable}

The proof is in the Appendix \ref{appendix_dynamics}.
As for a single clause, we set the dynamics of the entire LagONN system to reach a saddle point as:
\begin{equation}
    \begin{cases}
        \tau \dot{\phi_x}=-\nabla_{\phi_x} L_T=-\sum_{m} \Vec{u_\lambda^m}.\partial \Vec{Z_m}/\partial \phi_x \\
        \tau_\lambda\dot{\phi_\lambda}=+\nabla_{\phi_\lambda} L_T=\sum_{m}\Vec{u_\lambda'^m}.\Vec{Z_m}
    \end{cases}
    \label{LagONN_dynamics}
\end{equation}
Fig.\ref{LagONN_sim}a and b show examples of phase dynamics without and with Lagrange oscillators for $N=20$ variables and $M=91$ clauses (satisfiable instance cnf-20-01 from SATlib \cite{SATlib}). Without Lagrange oscillators, the ONN phases converge to non-binary values. To extract Boolean assignments, we round each phase to the nearest multiple of $\pi$, which yields a sub-optimal solution with two unsatisfied clauses. With Lagrange oscillators, the dynamics become more complex and take longer to settle. However, phases settle to multiples of $\pi$, eliminating the need for rounding, and providing an optimal Boolean assignment where all 91 clauses are satisfied. Since each clause enforces two constraints, $\text{Re}[Z]=0$, and $\text{Im}[Z]=0$, the system is overdetermined with $2M$ equations and only $N$ unknowns ($M>N$). We hypothesize that this high level of frustration restricts the saddle points to binary phases only. 

Fig.\ref{LagONN_sim}d shows the dynamics of the first four energy terms $Z_m$, which are simultaneously attracted to the origin $Z_m=0$ corresponding to a target saddle point. An optimal assignment is obtained when all $Z_m$ trajectories converge to zero.
Fig.\ref{LagONN_sim}c compares the cost function evolution for the two cases. Interestingly, LagONN dynamics reduce the cost as quickly as the ONN, reaching about two unsatisfied clauses after a single oscillation cycle. However, unlike the ONN that gets trapped, LagONN continues to evolve and eventually finds an optimal Boolean assignment. Comparing the phase dynamics from Fig.\ref{LagONN_sim}b with the cost evolution in Fig.\ref{LagONN_sim}c reveals that some phases keep evolving with little effect on the cost.
For this reason, in all the simulations reported in this paper, we monitor the cost in real-time and stop the simulation once the cost reaches a satisfactory value (set to 0 for satisfiable instances), as described in the Appendix \ref{appendix_cost_function}. In other words, we do not require full convergence to the saddle point shown in Fig.\ref{LagONN_sim}b, since its stability under the dynamics of Eq. \ref{LagONN_dynamics} is not guaranteed (see Appendix \ref{appendix_stability}).

\begin{figure}[t!]
\centering
\includegraphics[width=\textwidth]{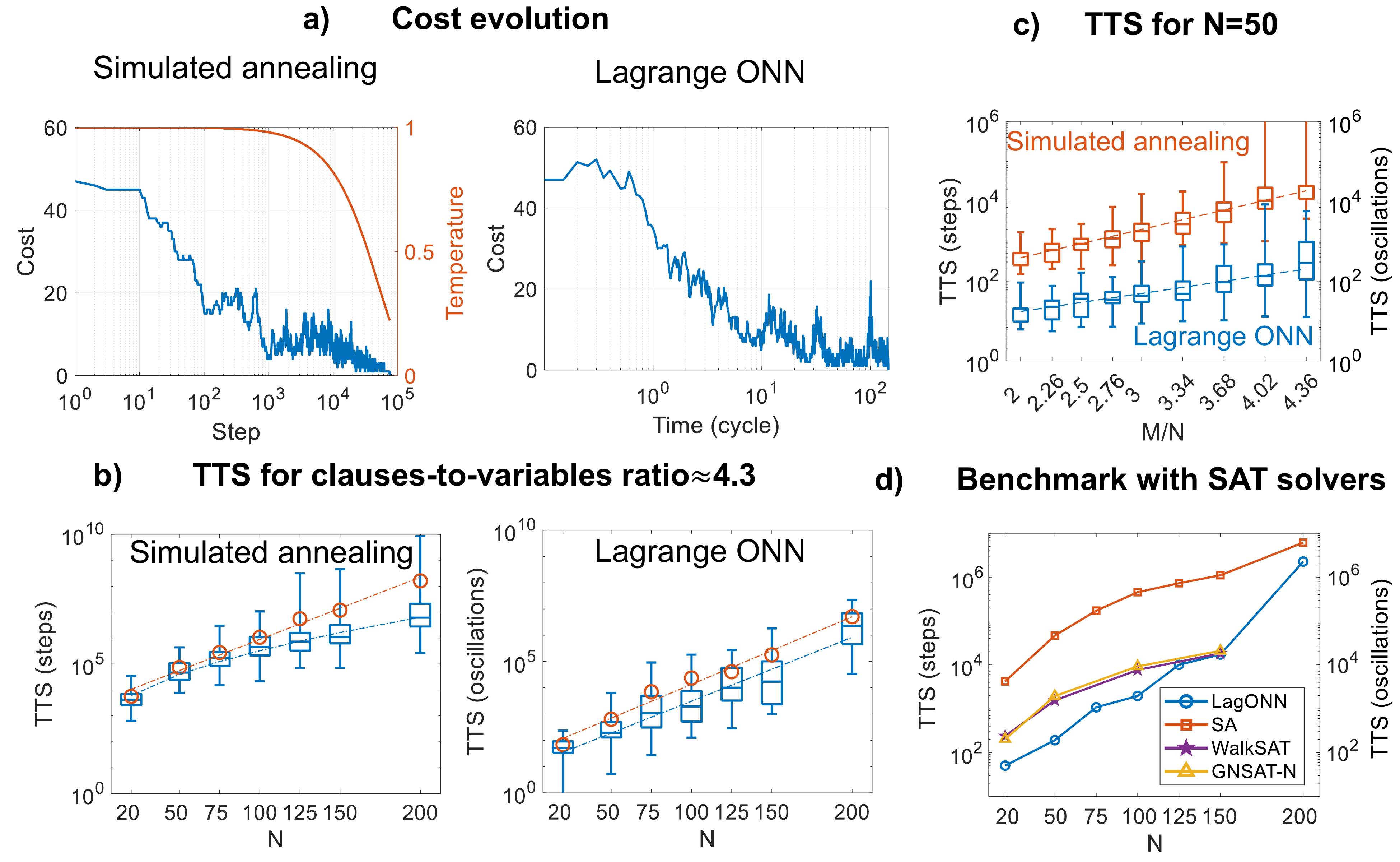}
\caption{a) Cost evolution comparison for a SATlib instance with 100 variables and 430 clauses, where all variables are initialized to '1'. The temperature in simulated annealing decreases exponentially with the number of steps as described in the Appendix \ref{appendix_SA}. b) Estimation of the time to reach an optimal solution (TTS) with 99\% probability for SATlib instances with a clauses-to-variables ratio $M/N\approx 4.3$. Boxes show the 1st, 2nd, and 3rd quartiles computed for the first 30 satisfiable instances from SATlib for LagONN, and 100 instances for simulated annealing. The error bars show the min-max values, and red circles are the averages. 
For both methods, we fit the logarithm of data to estimate the mean TTS scaling as $\sim \exp(aN+b)$ with $a=0.059$, $b=3.58$ for LagONN and $a=0.055$, $b=8.16$ for simulated annealing (red dashed lines). Simulated annealing has an advantageous scaling for the median TTS fitted as $\sim \exp(c\sqrt{N}+d)$ with $c=0.719$, $d=5.49$, whereas we fit LagONN's median TTS as  $\sim \exp(aN+b)$ with $a=0.056$ and $b=2.43$ (blue dashed curves). c) For a fixed number of variables $N=50$, we vary the number of clauses $M$ from 100 to 218 and compute TTS for 100 instances per point. The median TTS scales exponentially with the number of clauses for both methods. d) Benchmark of median TTS with SAT solvers from Ref. \cite{Pedretti_2025} based on stochastic local search for satisfiable SATlib instances.}
\label{results_scaling}
\end{figure}

\subsection{Comparison with simulated annealing and time-to-solution}
We compare LagONN's search with a simulated annealing algorithm for SAT \cite{johnson_1996} detailed in the Appendix \ref{appendix_SA}. Each SAT variable is sequentially flipped and the corresponding cost change $\delta$ is calculated. A flip is accepted with probability $1/(1+\exp(\delta/T))$ where $T$ is a temperature parameter. The temperature decreases exponentially from $T_{max}=1$ to $T_{min}=0.01$, values chosen empirically by inspecting cost trajectories across different instance sizes. Following the exponential schedule proposed in \cite{johnson_1996}, we obtain satisfactory results, though we do not claim optimality.

An example of simulated annealing run is shown in Fig.\ref{results_scaling}a (left) for 100 variables and 430 clauses, with the x-axis indicating the number of steps (tentative flips). The 100 variables are arbitrarily initialized to '1', and after the first 100 steps (one sweep), each variable has been considered once and potentially flipped according to the sigmoid probability at $T=T_{max}$. This initial sweep induces a sharp cost drop of about 30 satisfied clauses. As the temperature decreases, the algorithm enters a regime with fewer flips and the cost gradually declines. The run terminates once an optimal variable assignment is found.

A LagONN simulation is shown in Fig.\ref{results_scaling}a (right) for the same variable initialization and a random Lagrange oscillator initialization. The cost trajectory differs qualitatively: unlike simulated annealing, LagONN has no annealing schedule and continues making uphill moves even at low cost, which simulated annealing would rarely accept at low temperature. Crucially, LagONN dynamics do not halt until all constraints are satisfied, whereas simulated annealing can freeze into a suboptimal state within finite annealing time. However, this does not imply greater efficiency, as shown in the next time-to-solution comparison.

We now compare LagONN and simulated annealing on hard random instances from SATlib \cite{SATlib} with increasing sizes $(N,M)\in\{(20,91),(50,218),(75,325),(100,430),(125,538),(150,645),(200,860)\}$.
These instances are part of a well-known benchmark, lying near the computational phase transition where the clause-to-variable ratio is $M/N\approx 4.3$ \cite{Mitchell_1992}, making them especially challenging. For each $(N,M)$, we evaluate the first 100 instances for simulated annealing and the first 30 for LagONN (fewer instances due to the longer runtime of LagONN simulations). Each instance is run with 100 random initializations, where the phases are drawn uniformly, and simulations are executed for a fixed runtime $t_{max}$. From these runs, we estimate the success probability $p_s$. As in other Ising-machine benchmarks, we report the time to solution (TTS), defined as the expected time to reach an optimal solution with probability 0.99. For each instance, the TTS is given by:
\begin{equation}
    TTS=t_{max}\frac{\log(0.01)}{\log(1-p_s)}
    \label{eq:timetosolution}
\end{equation}
We apply the same TTS formula for simulated annealing, where $t_{max}=n_{max}$ corresponds to the number of algorithmic steps per run (maximum number of variable updates). When analyzing TTS scaling with system size, it is crucial to optimize the runtime for each size. Otherwise, misleading effects may appear such as artifically flat curves when $t_{max}$ is set too large, or TTS overestimation when $t_{max}$ is too small for larger instances \cite{ronnow_2014}.
To mitigate these issues, we iteratively increase $t_{max}$ and $n_{max}$ with the system size, up to $t_{max}\leq 10^6$ oscillation cycles, and $n_{max}\leq 2\times 10^7$ steps. This ensures that the success probability remains within the range $0.1<p_s<0.9$ for all tested instances.

Fig.\ref{results_scaling}b shows the TTS scaling with the number of variables at a fixed clauses-to-variables ratio $M/N\approx4.3$, plotted on a semilog scale with blue boxes corresponding to 1st, 2nd, and 3rd quartiles. Qualitatively, simulated annealing exhibits more favorable scaling than LagONN on this benchmark, as its slope flattens at larger sizes. Specifically, we fit the median TTS as $\sim \exp(c\sqrt{N}+d)$ with $c=0.719$, $d=5.49$, whereas LagONN's median TTS follows a straight line in the semilog plot, fitted as $\sim \exp(aN+b)$ with $a=0.056$ and $b=2.43$ (blue dashed curves). For simulated annealing, we excluded a polynomial scaling for the median TTS, as the log-log plot did not yield a straight line.

At large sizes, LagONN's TTS has fewer outliers than simulated annealing, with a mean TTS scaling similarly to the median as $\sim \exp(aN+b)$ with $a=0.059$, $b=3.58$. In contrast, simulated annealing sometimes encounters very hard instances, leading to more than five decades of variations at $N=200$, which pushes the mean scaling to $\sim \exp(aN+b)$ with $a=0.055$, $b=8.16$, comparable in slope to LagONN. Overall, while simulated annealing achieves a more favorable median scaling, LagONN remains competitive due to its lower prefactor $\exp(2.43)<\exp(5.49)$, which keeps it competitive up to $N\leq200$.

As shown in Fig.~\ref{results_scaling}c for $N=50$, both algorithms exhibit an exponential increase in TTS as the clauses-to-variables ratio approaches the computational phase transition at $M/N \approx 4.3$. Between $M/N = 2$ and $4.3$, the TTS increases by roughly two orders of magnitude for simulated annealing and by about one order of magnitude for LagONN, confirming that problem difficulty grows with the number of clauses.
Fixing $M/N = 2$ for all system sizes halves LagONN’s median scaling slope in Fig.~\ref{results_scaling}b, yielding a new slope of $a = 0.027$ and a median TTS of $4.4 \times 10^3$ oscillations for $N = 200$.
We remain cautious about these scaling estimates, as they are based on a limited number of instances and sizes, constrained by the very long simulation times at $M/N\approx 4.3$ (see Appendix \ref{appendix_solver}). Running the same TTS scaling experiment at $M/N\approx4.3$ with a reduced integration time step (linearly annealed from 0.15 to 0.015) did not significantly improve LagONN's exponential scaling.

\subsection{Benchmark with SAT solvers}

We benchmark LagONN with SAT solvers based on stochastic local search (SLS), namely WalkSAT \cite{selman_1994,Hoos_2000} and a hardware-enhanced version of GSAT, GNSAT-N \cite{Pedretti_2025}. SLS algorithms flip variables iteratively, similar to simulated annealing, but differ in how they select variables. One of the simplest methods is GSAT \cite{selman_1992}, a greedy algorithm that flips a variable which maximizes the overall cost reduction or gain. Ref. \cite{Pedretti_2025} proposes GNSAT-N, a GSAT hardware acceleration and improvement by adding normal-distributed noise to the gain, thereby allowing random walks.
WalkSAT has another strategy to select variables and focuses on unsatisfied clauses. It first selects an unsatisfied clause and computes the number of new unsatisfied clauses (called breaks) induced when flipping each variable of the clause. If for some variable of the clause, break=0 (zero-damage flip \cite{Hoos_2000}), it is flipped. Otherwise, the best variable (minimizing breaks) is flipped with some probability, or a random variable is flipped (random walk). It is then a compromise between greedy moves and a random walk, set by the probability parameter.

Fig.\ref{results_scaling}d shows the median TTS (steps) for WalkSAT and GNSAT-N reported by Ref. \cite{Pedretti_2025} on uniform satisfiable SATlib instances, alongside LagONN's TTS expressed in number of oscillations. WalkSAT and GNSAT-N have similar TTS and exhibit a clear advantage in scaling over simulated annealing and LagONN. We interpret the performance gap between SAT solvers and simulated annealing as follows. Although the probability of flipping a variable in simulated annealing depends on the cost reduction (following a sigmoidal function), its current implementation (SASAT \cite{johnson_1996}) evaluates each variable iteratively in a fixed order. This contrasts with GSAT or WalkSAT, which have more sophisticated mechanisms to select variables or clauses before making a flip, potentially making more efficient moves.

Although LagONN also has a greedy mechanism to follow steepest trajectories minimizing the cost (gradient descent), its dynamics are strongly perturbed by the competitive forces introduced by the Lagrange oscillators, as shown in Fig.\ref{LagONN_principle}. This interplay of simultaneous gradient descent and ascent appears less efficient than the highly tuned SLS searches designed for SAT, which is reflected in LagONN’s less favorable TTS scaling. 
Nevertheless, its lower prefactor keeps LagONN competitive for $N\leq 150$, and the results overall confirm its ability to enforce constraint satisfaction by escaping infeasible states in a deterministic manner. More generally, these findings suggest that similar competitive oscillator dynamics could be applied to constrained optimization problems formulated through a Lagrange function (Eq. \ref{lagrange_function}) as we discuss next with the example of phase copying. In this case, Lagrange oscillators are used to enforce the constraints during optimization rather than serving purely as a search mechanism, as in Max-3-SAT.

\section{Summary and Discussion}
\begin{table}[t!]
\centering
\caption{Comparison between simulated annealing and Lagrange ONN for solving the Max-3-SAT constraint satisfaction problem, from both conceptual and practical (hardware) points of view.}
\label{table_discussion}
\resizebox{0.75\textwidth}{!}{%
\begin{tabular}{|c|c|c|}
\hline
&
  Simulated annealing &
  Lagrange ONN \\ \hline
Strengths &
  \begin{tabular}[c]{@{}c@{}}Advantageous TTS scaling.\\ Convergence at cold temperature.\end{tabular} &
  \begin{tabular}[c]{@{}c@{}}Noiseless approach.\\ Cannot settle into infeasible states.\end{tabular} \\ \hline
Weaknesses &
  \begin{tabular}[c]{@{}c@{}}Requires careful noise tuning.\\ Can be stuck at infeasible states.\end{tabular} &
  \begin{tabular}[c]{@{}c@{}}Seemingly worse TTS scaling.\\ Stability is not guaranteed.\end{tabular} \\ \hline
\end{tabular}%
}
\end{table}
In this work, we demonstrated how introducing additional Lagrange oscillators into coupled oscillator systems enables constraint satisfaction, exemplified by the Max-3-SAT problem. These problems represent special cases where the Lagrange function reduces to the constraint function $g$, which counts unsatisfied clauses, without any additional cost term $f$. Because the Lagrange oscillators drive the system until all constraints are satisfied, LagONN continues its search until it finds an optimal solution, a distinct advantage over simulated annealing, which can become trapped in local minima. Another benefit is that LagONN is noiseless and requires no annealing schedule, a practical advantage given that tuning noise in physical oscillator implementations can be challenging. On the other hand, LagONN’s search for optimal Boolean assignments is less effective than simulated annealing and state-of-the-art SAT solvers, as indicated by our time-to-solution analysis. Furthermore, the stability of the saddle points reached by the dynamics in Eq. \ref{LagONN_dynamics} is not guaranteed, although stabilization mechanisms are possible (see Appendix \ref{appendix_stability}). A side-by-side comparison of the strengths and limitations of each method is provided in Table \ref{table_discussion}.

Beyond pure constraint satisfaction problems, LagONN can also be applied to more general constrained problems involving both a cost $f$ and constraints $g$. For instance, Fig.\ref{other_constrained_problems} illustrates the problem of phase copying, which can arise in large-scale hardware implementations. When a problem is represented by a dense graph that cannot be directly implemented due to the quadratic number of edges, one can introduce copy nodes that copy each other to distribute the edges \cite{sajeeb_2025}. This scenario is typical for quantum annealers with limited connectivity \cite{pelofske_2025}. Here, the constraint $g$ enforces equality between copies, while the cost function $f$ corresponds to the Ising energy of the original graph. To copy the spin values, one can introduce a ferromagnetic coupling $J_c$ between nodes, corresponding to the penalty term $\frac{J_c}{2}(S_1-S_2)^2 \equiv -J_cS_1S_2 $. In principle, this positive coupling should be large enough to satisfy the constraint. However, in practice, too strong couplings can rigidify the dynamics and hinder the search for ground states \cite{Venturelli_2015,sajeeb_2025}. 

Fig.\ref{other_constrained_problems}a shows an example of a coupled oscillator network with standard coupling $J=-1$ (blue edges) and copy coupling $J_c$ (dashed black edges). The cost function $f$ to minimize is the XY Ising energy $E$ for phase oscillators (Eq. \ref{ONN_energy_eq}), determined by the symmetric couplings. The copy constraints are defined as $g=(\cos\phi_1-\cos\phi_2,\cos\phi_3-\cos\phi_4)=0$. When the constraints are satisfied, the graph reduces to a 4-node fully-connected graph. Fig.\ref{other_constrained_problems}b shows an example of phase dynamics with a weak copy coupling $J_c=+0.5$ where the constraints are not fully enforced, e.g. $\cos\phi_1\neq \cos\phi_2$.

\begin{figure}[t!]
\centering
\includegraphics[width=\textwidth]{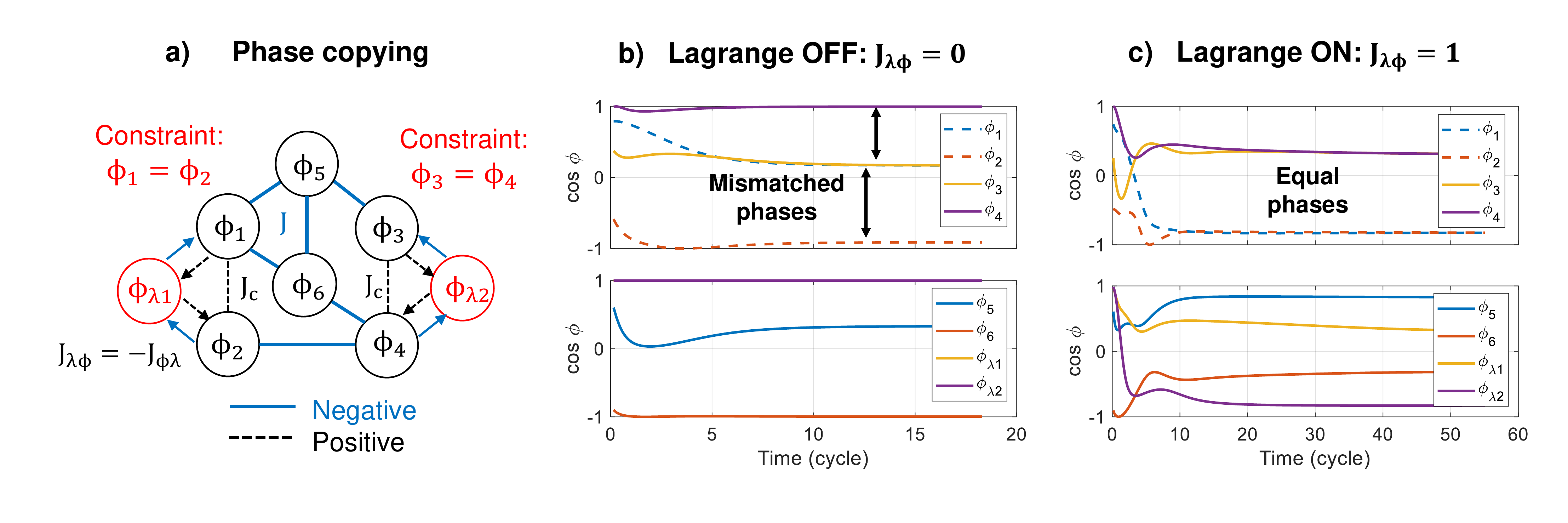}
\caption{LagONN for phase copying on an initial graph with negative coupling $J=-1$ (blue edge). a) In practice at a large scale, Ising machines often have limited hardware connectivity, and dense graphs are mapped to a sparser network by introducing copy nodes \cite{sajeeb_2025}. Ideally, the copy nodes should copy the original values, e.g. $\cos\phi_2=\cos\phi_1$. This is generally achieved with a ferromagnetic coupling between the two nodes ($J_c=+0.5$ here). In theory, the positive coupling should be large enough to keep the two variables identical. However, a large coupling value introduces rigidity in the system, hindering the ground state search. b) For weak coupling values such as in this example, the copied phases are not equal ($\cos\phi_1\neq \cos\phi_2$). c) Lagrange oscillators can enforce the copy constraints between two copied nodes, as shown in this example for $J_{\lambda \phi}=1$.}
\label{other_constrained_problems}
\end{figure}

Instead of increasing $J_c$ at the risk of rigidifying the system, one can introduce a Lagrange oscillator for each constraint to enforce phase equality, as illustrated in Fig.\ref{other_constrained_problems}c. For example, for the phases $\phi_1$ and $\phi_2$, the corresponding copy constraint can be written as a vector $\Vec{Z_{1-2}}=\exp(i\phi_1)-\exp(i\phi_2)=0$.
Analogous to the Lagrange function defined for 3-SAT (Eq. \ref{lagoon_Lagrange_function}), the Lagrange term for this constraint is expressed as:
\begin{equation}
     \Vec{u_{\lambda1}}.\Vec{Z_{1-2}}=\cos(\phi_1-\phi_{\lambda1})-\cos(\phi_2-\phi_{\lambda1})
\end{equation}
where $\phi_{\lambda1}$ is the phase of the Lagrange oscillator connected to the first pair of oscillators. The total Lagrange function for this example is then:
\begin{equation}
    L(\phi,\phi_\lambda)=E(\phi)+J_{\lambda \phi}\big(\Vec{u_{\lambda1}}.\Vec{Z_{1-2}}+\Vec{u_{\lambda2}}.\Vec{Z_{3-4}}\big)
\end{equation}

where $J_{\lambda \phi}$ sets the constraint strength or equivalently, the speed of Lagrange oscillators. Implementing competitive dynamics seeking a saddle point in the Lagrange landscape as in Eq. \ref{LagONN_dynamics}, $\phi_\lambda$ performs gradient ascent until the constraints are satisfied, while other phases $\phi$ follow gradient descent to minimize $L$. The opposite gradient signs induce asymmetric couplings between oscillators and Lagrange oscillators, illustrated by the arrows in Fig.\ref{other_constrained_problems}a with weights set to $J_{\lambda \phi}=-J_{\phi \lambda}=1$.

A key limitation of this approach is that the stability of fixed points is not guaranteed and requires further analysis. For example, in simulations, increasing the strength $J_{\lambda \phi}$ speeds up the Lagrange oscillators and introduces transient and decaying phase oscillations. In contrast, increasing the copying strength $J_c$ acts like additional damping and smooths the dynamics. Combining energy penalties and Lagrange multipliers could enhance existing oscillatory Ising machines, as proposed in previous work for knapsack problems \cite{delacour_2025}, potentially yielding orders-of-magnitude speed-ups. A detailed exploration of these regimes is left for future work.

\section{Conclusion}
This article introduced LagONN, a Lagrange oscillatory neural network that enforces constraint satisfaction as demonstrated for the Max-3-SAT problem. Unlike gradient-descent-based approaches such as oscillatory Ising machines, which can be trapped in infeasible local minima, LagONN employs additional Lagrange oscillators to ensure that 3-SAT clauses are satisfied. Conceptually, these Lagrange variables provide alternative pathways in the energy landscape to escape local minima and reach optimal states where phases correspond to optimal Boolean values. When benchmarked on SATlib instances up to 200 variables and 860 clauses against simulated annealing and SAT solvers, LagONN simulations exhibited less favorable time-to-solution scaling with increasing problem size, yet remained competitive for the tested problem sizes ($N\leq 200$) due to a smaller prefactor in the scaling function. For the Max-3-SAT constraint satisfaction problem, LagONN offers a deterministic search method that eliminates the need for careful noise control required in simulated annealing, which can be difficult to implement in hardware. Furthermore, the example of phase copying illustrates how oscillatory-based Ising machines augmented with Lagrange oscillators can enforce constraint satisfaction in general optimization tasks, going beyond traditional penalty methods.

\section*{Declarations}

This work was supported by the European Union’s Horizon 2020 research
and innovation program, EU H2020 NEURONN (www.neuronn.eu) project
under Grant 871501.
BM, FS, and ATS acknowledge support from the European Union's Horizon Europe research and innovation programme, PHASTRAC project under grant agreement No 101092096 and Dutch Research Council‘s AiNed Fellowship research programme, AI-on-ONN project under grant agreement No. NGF.1607.22.016, as well as funding from the European Research Council ERC THERMODON project under grant agreement No. 101125031.

\section*{Code availability}

Matlab codes and data are available at the following GitHub repository: 

https://github.com/corentindelacour/Lagrange-oscillatory-neural-network







\bibliographystyle{unsrt}  
\bibliography{bib}

\begin{thebibliography}{10}

\bibitem{Cook_1971}
Stephen~A. Cook.
\newblock The complexity of theorem-proving procedures.
\newblock In {\em Proceedings of the Third Annual ACM Symposium on Theory of Computing}, STOC '71, page 151–158, New York, NY, USA, 1971. Association for Computing Machinery.

\bibitem{Mohseni_2022}
Naeimeh Mohseni, Peter~L. McMahon, and Tim Byrnes.
\newblock Ising machines as hardware solvers of combinatorial optimization problems.
\newblock {\em Nature Reviews Physics}, 4(6):363--379, Jun 2022.

\bibitem{Hopfiel_1982}
J~J Hopfield.
\newblock Neural networks and physical systems with emergent collective computational abilities.
\newblock {\em Proceedings of the National Academy of Sciences}, 79(8):2554--2558, 1982.

\bibitem{Tank_1986}
D.~Tank and J.~Hopfield.
\newblock Simple 'neural' optimization networks: {An} {A}/{D} converter, signal decision circuit, and a linear programming circuit.
\newblock {\em IEEE Transactions on Circuits and Systems}, 33(5):533--541, May 1986.

\bibitem{Hopfield_1984}
J~J Hopfield.
\newblock Neurons with graded response have collective computational properties like those of two-state neurons.
\newblock {\em Proceedings of the National Academy of Sciences}, 81(10):3088--3092, May 1984.

\bibitem{Lucas_2014}
Andrew Lucas.
\newblock Ising formulations of many {NP} problems.
\newblock {\em Frontiers in Physics}, 2, 2014.

\bibitem{Venturelli_2015}
Davide Venturelli, Salvatore Mandr\`a, Sergey Knysh, Bryan O'Gorman, Rupak Biswas, and Vadim Smelyanskiy.
\newblock Quantum optimization of fully connected spin glasses.
\newblock {\em Phys. Rev. X}, 5:031040, Sep 2015.

\bibitem{parizy_2021}
Matthieu Parizy and Nozomu Togawa.
\newblock Analysis and {Acceleration} of the {Quadratic} {Knapsack} {Problem} on an {Ising} {Machine}.
\newblock {\em IEICE Transactions on Fundamentals of Electronics, Communications and Computer Sciences}, E104.A(11):1526--1535, November 2021.

\bibitem{cellini_2024}
Lorenzo Cellini, Antonio Macaluso, and Michele Lombardi.
\newblock {QAL}-{BP}: an augmented {Lagrangian} quantum approach for bin packing.
\newblock {\em Scientific Reports}, 14(1):5142, March 2024.

\bibitem{hestenes_1969}
Magnus~R. Hestenes.
\newblock Multiplier and gradient methods.
\newblock {\em Journal of Optimization Theory and Applications}, 4(5):303--320, November 1969.

\bibitem{powell_1978}
M.~J.~D. Powell.
\newblock Algorithms for nonlinear constraints that use lagrangian functions.
\newblock {\em Mathematical Programming}, 14(1):224--248, December 1978.

\bibitem{vadlamani_2020}
Sri~Krishna Vadlamani, Tianyao~Patrick Xiao, and Eli Yablonovitch.
\newblock Physics successfully implements {Lagrange} multiplier optimization.
\newblock {\em Proceedings of the National Academy of Sciences}, 117(43):26639--26650, October 2020.

\bibitem{Nagamatu_1996}
M.~Nagamatu and T.~Yanaru.
\newblock On the stability of lagrange programming neural networks for satisfiability problems of prepositional calculus.
\newblock {\em Neurocomputing}, 13(2):119--133, 1996.
\newblock Soft Computing.

\bibitem{Kirkpatrick_1983}
S.~Kirkpatrick, C.~D. Gelatt, and M.~P. Vecchi.
\newblock Optimization by simulated annealing.
\newblock {\em Science}, 220(4598):671--680, 1983.

\bibitem{Todri_2024}
Aida Todri-Sanial, Corentin Delacour, Madeleine Abernot, and Filip Sabo.
\newblock Computing with oscillators from theoretical underpinnings to applications and demonstrators.
\newblock {\em npj Unconventional Computing}, 1(1):14, Dec 2024.

\bibitem{Parihar_2017}
Abhinav Parihar, Nikhil Shukla, Matthew Jerry, Suman Datta, and Arijit Raychowdhury.
\newblock Vertex coloring of graphs via phase dynamics of coupled oscillatory networks.
\newblock {\em Scientific Reports}, 7(1):911, Apr 2017.

\bibitem{Wang_2021}
Tianshi Wang, Leon Wu, Parth Nobel, and Jaijeet Roychowdhury.
\newblock Solving combinatorial optimisation problems using oscillator based ising machines.
\newblock {\em Natural Computing}, 20(2):287--306, Jun 2021.

\bibitem{Graber_2022}
Markus Graber and Klaus Hofmann.
\newblock A versatile and adjustable 400 node cmos oscillator based ising machine to investigate and optimize the internal computing principle.
\newblock In {\em 2022 IEEE 35th International System-on-Chip Conference (SOCC)}, pages 1--6, 2022.

\bibitem{Graber_2023}
Markus Graber and Klaus Hofmann.
\newblock A {Coupled} {Oscillator} {Network} to {Solve} {Combinatorial} {Optimization} {Problems} with {Over} 95\% {Accuracy}.
\newblock In {\em 2023 {IEEE} {International} {Symposium} on {Circuits} and {Systems} ({ISCAS})}, pages 1--5, Monterey, CA, USA, May 2023. IEEE.

\bibitem{Delacour_2023}
Corentin Delacour, Stefania Carapezzi, Gabriele Boschetto, Madeleine Abernot, Thierry Gil, Nadine Azemard, and Aida Todri-Sanial.
\newblock A mixed-signal oscillatory neural network for scalable analog computations in phase domain.
\newblock {\em Neuromorphic Computing and Engineering}, 3(3):034004, aug 2023.

\bibitem{Bashar_2023_high_order}
Mohammad~Khairul Bashar and Nikhil Shukla.
\newblock Designing {Ising} machines with higher order spin interactions and their application in solving combinatorial optimization.
\newblock {\em Scientific Reports}, 13(1):9558, June 2023.

\bibitem{Mallick_2020}
Antik Mallick, Mohammad~Khairul Bashar, Daniel~S. Truesdell, Benton~H. Calhoun, Siddharth Joshi, and Nikhil Shukla.
\newblock Using synchronized oscillators to compute the maximum independent set.
\newblock {\em Nature Communications}, 11(1):4689, Sep 2020.

\bibitem{Ahmed_2021}
Ibrahim Ahmed, Po-Wei Chiu, William Moy, and Chris~H. Kim.
\newblock A probabilistic compute fabric based on coupled ring oscillators for solving combinatorial optimization problems.
\newblock {\em IEEE Journal of Solid-State Circuits}, 56(9):2870--2880, 2021.

\bibitem{Moy_2022}
William Moy, Ibrahim Ahmed, Po-wei Chiu, John Moy, Sachin~S. Sapatnekar, and Chris~H. Kim.
\newblock A 1,968-node coupled ring oscillator circuit for combinatorial optimization problem solving.
\newblock {\em Nature Electronics}, 5(5):310--317, May 2022.

\bibitem{Graber_2024}
Markus Graber and Klaus Hofmann.
\newblock An integrated coupled oscillator network to solve optimization problems.
\newblock {\em Communications Engineering}, 3(1):116, Aug 2024.

\bibitem{Grollier_2020}
J.~Grollier, D.~Querlioz, K.~Y. Camsari, K.~Everschor-Sitte, S.~Fukami, and M.~D. Stiles.
\newblock Neuromorphic spintronics.
\newblock {\em Nature Electronics}, 3(7):360--370, Jul 2020.

\bibitem{Grimaldi_2023}
Andrea Grimaldi, Luciano Mazza, Eleonora Raimondo, Pietro Tullo, Davi Rodrigues, Kerem~Y. Camsari, Vincenza Crupi, Mario Carpentieri, Vito Puliafito, and Giovanni Finocchio.
\newblock Evaluating spintronics-compatible implementations of ising machines.
\newblock {\em Phys. Rev. Appl.}, 20:024005, Aug 2023.

\bibitem{Dutta_2021}
S.~Dutta, A.~Khanna, A.~S. Assoa, H.~Paik, D.~G. Schlom, Z.~Toroczkai, A.~Raychowdhury, and S.~Datta.
\newblock An ising hamiltonian solver based on coupled stochastic phase-transition nano-oscillators.
\newblock {\em Nature Electronics}, 4(7):502--512, Jul 2021.

\bibitem{Kim_2023}
Hyun~Wook Kim, Seonuk Jeon, Heebum Kang, Eunryeong Hong, Nayeon Kim, and Jiyong Woo.
\newblock Understanding rhythmic synchronization of oscillatory neural networks based on nbox artificial neurons for edge detection.
\newblock {\em IEEE Transactions on Electron Devices}, 70(6):3031--3036, 2023.

\bibitem{Maher_2024}
Olivier Maher, Manuel Jim{\'e}nez, Corentin Delacour, Nele Harnack, Juan N{\'u}{\~{n}}ez, Mar{\'i}a~J. Avedillo, Bernab{\'e} Linares-Barranco, Aida Todri-Sanial, Giacomo Indiveri, and Siegfried Karg.
\newblock A cmos-compatible oscillation-based vo2 ising machine solver.
\newblock {\em Nature Communications}, 15(1):3334, Apr 2024.

\bibitem{Yun_2024}
Seong-Yun Yun, Joon-Kyu Han, and Yang-Kyu Choi.
\newblock A nanoscale bistable resistor for an oscillatory neural network.
\newblock {\em Nano Letters}, 24(9):2751--2757, Mar 2024.

\bibitem{Sharma_2023}
Anshujit Sharma, Matthew Burns, Andrew Hahn, and Michael Huang.
\newblock Augmenting an electronic ising machine to effectively solve boolean satisfiability.
\newblock {\em Scientific Reports}, 13(1):22858, Dec 2023.

\bibitem{Su_2023}
Yuqi Su, Tony Tae-Hyoung Kim, and Bongjin Kim.
\newblock A reconfigurable lsing machine for boolean satisfiability problems featuring many-body spin interactions.
\newblock In {\em 2023 IEEE Custom Integrated Circuits Conference (CICC)}, pages 1--2, 2023.

\bibitem{Hizzani_2024}
Mohammad Hizzani, Arne Heittmann, George Hutchinson, Dmitrii Dobrynin, Thomas Van~Vaerenbergh, Tinish Bhattacharya, Adrien Renaudineau, Dmitri Strukov, and John~Paul Strachan.
\newblock Memristor-based hardware and algorithms for higher-order hopfield optimization solver outperforming quadratic ising machines.
\newblock In {\em 2024 IEEE International Symposium on Circuits and Systems (ISCAS)}, pages 1--5, 2024.

\bibitem{Pedretti_2025}
Giacomo Pedretti, Fabian B{\"o}hm, Tinish Bhattacharya, Arne Heittmann, Xiangyi Zhang, Mohammad Hizzani, George Hutchinson, Dongseok Kwon, John Moon, Elisabetta Valiante, Ignacio Rozada, Catherine~E. Graves, Jim Ignowski, Masoud Mohseni, John~Paul Strachan, Dmitri Strukov, Ray Beausoleil, and Thomas Van~Vaerenbergh.
\newblock Solving boolean satisfiability problems with resistive content addressable memories.
\newblock {\em npj Unconventional Computing}, 2(1):7, Apr 2025.

\bibitem{Dikopoulos_2025}
Evangelos Dikopoulos, Ying–Tuan Hsu, Luke Wormald, Wei Tang, Zhengya Zhang, and Michael~P. Flynn.
\newblock 25.1 a physics-inspired oscillator-based mixed-signal optimization engine for solving 50-variable 218-clause 3-sat problems with 100
\newblock In {\em 2025 IEEE International Solid-State Circuits Conference (ISSCC)}, volume~68, pages 01--03, 2025.

\bibitem{Salim_2025}
Ahmet~Yusuf Salim, Bart Selman, Henry Kautz, Zeljko Ignjatovic, and Selçuk Köse.
\newblock Ski-sat: A cmos-compatible hardware for solving sat problems.
\newblock {\em IEEE Transactions on Circuits and Systems I: Regular Papers}, pages 1--12, 2025.

\bibitem{johnson_1996}
David Johnson and Michael Trick, editors.
\newblock {\em Cliques, {Coloring}, and {Satisfiability}}, volume~26 of {\em {DIMACS} {Series} in {Discrete} {Mathematics} and {Theoretical} {Computer} {Science}}.
\newblock American Mathematical Society, Providence, Rhode Island, October 1996.

\bibitem{selman_1994}
Bart Selman, Henry~A. Kautz, and Bram Cohen.
\newblock Noise strategies for improving local search.
\newblock In {\em Proceedings of the Twelfth National Conference on Artificial Intelligence (Vol. 1)}, AAAI '94, page 337–343, USA, 1994. American Association for Artificial Intelligence.

\bibitem{SATlib}
Holger Hoos.
\newblock Satlib - benchmark problems.
\newblock \url{https://www.cs.ubc.ca/~hoos/SATLIB/benchm.html}, 2000.
\newblock [Online; accessed 13-October 2023].

\bibitem{Camsari_2017}
Kerem~Yunus Camsari, Rafatul Faria, Brian~M. Sutton, and Supriyo Datta.
\newblock Stochastic $p$-bits for invertible logic.
\newblock {\em Phys. Rev. X}, 7:031014, Jul 2017.

\bibitem{Cai_2020}
Fuxi Cai, Suhas Kumar, Thomas Van~Vaerenbergh, Xia Sheng, Rui Liu, Can Li, Zhan Liu, Martin Foltin, Shimeng Yu, Qiangfei Xia, J.~Joshua Yang, Raymond Beausoleil, Wei~D. Lu, and John~Paul Strachan.
\newblock Power-efficient combinatorial optimization using intrinsic noise in memristor hopfield neural networks.
\newblock {\em Nature Electronics}, 3(7):409--418, Jul 2020.

\bibitem{chowdhury_2023}
Shuvro Chowdhury, Andrea Grimaldi, Navid~Anjum Aadit, Shaila Niazi, Masoud Mohseni, Shun Kanai, Hideo Ohno, Shunsuke Fukami, Luke Theogarajan, Giovanni Finocchio, Supriyo Datta, and Kerem~Y. Camsari.
\newblock A full-stack view of probabilistic computing with p-bits: Devices, architectures, and algorithms.
\newblock {\em IEEE Journal on Exploratory Solid-State Computational Devices and Circuits}, 9(1):1--11, 2023.

\bibitem{Singh_2024}
Nihal~Sanjay Singh, Keito Kobayashi, Qixuan Cao, Kemal Selcuk, Tianrui Hu, Shaila Niazi, Navid~Anjum Aadit, Shun Kanai, Hideo Ohno, Shunsuke Fukami, and Kerem~Y. Camsari.
\newblock Cmos plus stochastic nanomagnets enabling heterogeneous computers for probabilistic inference and learning.
\newblock {\em Nature Communications}, 15(1):2685, Mar 2024.

\bibitem{Aadit_2022}
Navid~Anjum Aadit, Andrea Grimaldi, Mario Carpentieri, Luke Theogarajan, John~M. Martinis, Giovanni Finocchio, and Kerem~Y. Camsari.
\newblock Massively {Parallel} {Probabilistic} {Computing} with {Sparse} {Ising} {Machines}.
\newblock {\em Nature Electronics}, 5(7):460--468, June 2022.
\newblock arXiv:2110.02481 [cond-mat].

\bibitem{Aarts_1989}
Emile Aarts and Jan Korst.
\newblock {\em Simulated annealing and Boltzmann machines: a stochastic approach to combinatorial optimization and neural computing}.
\newblock John Wiley \& Sons, Inc., USA, 1989.

\bibitem{fahimi_2021}
Z.~Fahimi, M.~R. Mahmoodi, H.~Nili, Valentin Polishchuk, and D.~B. Strukov.
\newblock Combinatorial optimization by weight annealing in memristive hopfield networks.
\newblock {\em Scientific Reports}, 11(1):16383, August 2021.

\bibitem{jiang_2023}
Mingrui Jiang, Keyi Shan, Chengping He, and Can Li.
\newblock Efficient combinatorial optimization by quantum-inspired parallel annealing in analogue memristor crossbar.
\newblock {\em Nature Communications}, 14(1):5927, September 2023.

\bibitem{Goto_2019}
Hayato Goto, Kosuke Tatsumura, and Alexander~R. Dixon.
\newblock Combinatorial optimization by simulating adiabatic bifurcations in nonlinear {Hamiltonian} systems.
\newblock {\em Science Advances}, 5(4):eaav2372, April 2019.

\bibitem{ansari_2024}
Md~Shabaz Ansari, Hemkant Nehete, Andrea Grimaldi, Luciano Mazza, Eleonora Raimondo, Vito Puliafito, Giovanni Finocchio, and Brajesh~Kumar Kaushik.
\newblock A comparison of oscillatory ising machines and simulated bifurcation machines for solving maximum cut problems.
\newblock In {\em 2024 IEEE 24th International Conference on Nanotechnology (NANO)}, pages 389--393, 2024.

\bibitem{Ercsey-ravasz_2011}
Mária Ercsey-Ravasz and Zoltán Toroczkai.
\newblock Optimization hardness as transient chaos in an analog approach to constraint satisfaction.
\newblock {\em Nature Physics}, 7(12):966--970, December 2011.

\bibitem{Molnar_2013}
Botond Molnár and Mária Ercsey-Ravasz.
\newblock Asymmetric {Continuous}-{Time} {Neural} {Networks} without {Local} {Traps} for {Solving} {Constraint} {Satisfaction} {Problems}.
\newblock {\em PLoS ONE}, 8(9):e73400, September 2013.

\bibitem{Traversa_2017}
Fabio~L. Traversa and Massimiliano Di~Ventra.
\newblock Polynomial-time solution of prime factorization and {NP}-complete problems with digital memcomputing machines.
\newblock {\em Chaos: An Interdisciplinary Journal of Nonlinear Science}, 27(2):023107, February 2017.

\bibitem{Suzuki_2013}
Hideyuki Suzuki, Jun-ichi Imura, Yoshihiko Horio, and Kazuyuki Aihara.
\newblock Chaotic boltzmann machines.
\newblock {\em Scientific Reports}, 3(1):1610, Apr 2013.

\bibitem{Lee_2025}
Kyle Lee, Shuvro Chowdhury, and Kerem~Y. Camsari.
\newblock Noise-augmented chaotic ising machines for combinatorial optimization and sampling.
\newblock {\em Communications Physics}, 8(1):35, Jan 2025.

\bibitem{Molnar_2018}
Botond Molnár, Ferenc Molnár, Melinda Varga, Zoltán Toroczkai, and Mária Ercsey-Ravasz.
\newblock A continuous-time {MaxSAT} solver with high analog performance.
\newblock {\em Nature Communications}, 9(1):4864, November 2018.

\bibitem{chang_2022}
Muya Chang, Xunzhao Yin, Zoltan Toroczkai, Xiaobo Hu, and Arijit Raychowdhury.
\newblock An analog clock-free compute fabric base on continuous-time dynamical system for solving combinatorial optimization problems.
\newblock In {\em 2022 IEEE Custom Integrated Circuits Conference (CICC)}, pages 1--2, 2022.

\bibitem{Izhikevich_2006}
E.M. Izhikevich and Y.~Kuramoto.
\newblock Weakly coupled oscillators.
\newblock In {\em Encyclopedia of Mathematical Physics}, pages 448--453. Academic Press, 2006.

\bibitem{Jackson_2018}
T.~{Jackson}, S.~{Pagliarini}, and L.~{Pileggi}.
\newblock An oscillatory neural network with programmable resistive synapses in 28 nm cmos.
\newblock In {\em 2018 IEEE International Conference on Rebooting Computing (ICRC)}, pages 1--7, 2018.

\bibitem{Abernot_2021}
Madeleine Abernot, Thierry Gil, Manuel Jiménez, Juan Núñez, María~J. Avellido, Bernabé Linares-Barranco, Théophile Gonos, Tanguy Hardelin, and Aida Todri-Sanial.
\newblock Digital implementation of oscillatory neural network for image recognition applications.
\newblock {\em Frontiers in Neuroscience}, 15:1095, 2021.

\bibitem{Wang_2019}
Tianshi Wang, Leon Wu, and Jaijeet Roychowdhury.
\newblock New computational results and hardware prototypes for oscillator-based ising machines.
\newblock In {\em Proceedings of the 56th Annual Design Automation Conference 2019}, DAC '19, New York, NY, USA, 2019. Association for Computing Machinery.

\bibitem{Burer_2001}
Samuel Burer, Renato Monteiro, and Yin Zhang.
\newblock Rank-two relaxation heuristics for max-cut and other binary quadratic programs.
\newblock {\em SIAM Journal on Optimization}, 12, 07 2001.

\bibitem{Erementchouk_2022}
Mikhail Erementchouk, Aditya Shukla, and Pinaki Mazumder.
\newblock On computational capabilities of ising machines based on nonlinear oscillators.
\newblock {\em Physica D: Nonlinear Phenomena}, 437:133334, 2022.

\bibitem{Bybee_2023}
Connor Bybee, Denis Kleyko, Dmitri~E. Nikonov, Amir Khosrowshahi, Bruno~A. Olshausen, and Friedrich~T. Sommer.
\newblock Efficient optimization with higher-order ising machines.
\newblock {\em Nature Communications}, 14(1):6033, September 2023.

\bibitem{fisher_2004}
Marshall~L. Fisher.
\newblock The {Lagrangian} {Relaxation} {Method} for {Solving} {Integer} {Programming} {Problems}.
\newblock {\em Management Science}, 50(12,):1861--1871, 2004.

\bibitem{Wah_1999}
Benjamin~W. Wah and Zhe Wu.
\newblock The {Theory} of {Discrete} {Lagrange} {Multipliers} for {Nonlinear} {Discrete} {Optimization}.
\newblock In Gerhard Goos, Juris Hartmanis, Jan Van~Leeuwen, and Joxan Jaffar, editors, {\em Principles and {Practice} of {Constraint} {Programming} – {CP}’99}, volume 1713, pages 28--42. Springer Berlin Heidelberg, Berlin, Heidelberg, 1999.
\newblock Series Title: Lecture Notes in Computer Science.

\bibitem{boyd_2023}
Stephen~P. Boyd and Lieven Vandenberghe.
\newblock {\em Convex optimization}.
\newblock Cambridge University Press, Cambridge New York Melbourne New Delhi Singapore, version 29 edition, 2023.

\bibitem{Zhang_1992}
S.~Zhang and A.G. Constantinides.
\newblock Lagrange programming neural networks.
\newblock {\em IEEE Transactions on Circuits and Systems II: Analog and Digital Signal Processing}, 39(7):441--452, July 1992.

\bibitem{Corti_2021}
Elisabetta Corti, Joaquin~Antonio Cornejo~Jimenez, Kham~M. Niang, John Robertson, Kirsten~E. Moselund, Bernd Gotsmann, Adrian~M. Ionescu, and Siegfried Karg.
\newblock Coupled vo2 oscillators circuit as analog first layer filter in convolutional neural networks.
\newblock {\em Frontiers in Neuroscience}, 15:19, 2021.

\bibitem{Mitchell_1992}
David Mitchell, Bart Selman, and Hector Levesque.
\newblock Hard and easy distributions of sat problems.
\newblock In {\em Proceedings of the Tenth National Conference on Artificial Intelligence}, AAAI'92, page 459–465. AAAI Press, 1992.

\bibitem{ronnow_2014}
Troels~F. Rønnow, Zhihui Wang, Joshua Job, Sergio Boixo, Sergei~V. Isakov, David Wecker, John~M. Martinis, Daniel~A. Lidar, and Matthias Troyer.
\newblock Defining and detecting quantum speedup.
\newblock {\em Science}, 345(6195):420--424, 2014.

\bibitem{Hoos_2000}
Holger~H. Hoos and Thomas St{\"u}tzle.
\newblock Local search algorithms for sat: An empirical evaluation.
\newblock {\em Journal of Automated Reasoning}, 24(4):421--481, May 2000.

\bibitem{selman_1992}
Bart Selman, Hector Levesque, and David Mitchell.
\newblock A new method for solving hard satisfiability problems.
\newblock In {\em Proceedings of the Tenth National Conference on Artificial Intelligence}, AAAI'92, page 440–446. AAAI Press, 1992.

\bibitem{sajeeb_2025}
M~Mahmudul~Hasan Sajeeb, Navid~Anjum Aadit, Shuvro Chowdhury, Tong Wu, Cesely Smith, Dhruv Chinmay, Atharva Raut, Kerem~Y. Camsari, Corentin Delacour, and Tathagata Srimani.
\newblock Scalable connectivity for ising machines: Dense to sparse.
\newblock {\em arXiv:2503.01177}, 2025.

\bibitem{pelofske_2025}
Elijah Pelofske.
\newblock Comparing three generations of d-wave quantum annealers for minor embedded combinatorial optimization problems.
\newblock {\em Quantum Science and Technology}, 10, 02 2025.

\bibitem{delacour_2025}
Corentin Delacour.
\newblock Self-adaptive ising machines for constrained optimization.
\newblock {\em arXiv:2501.04971}, 2025.

\bibitem{Bashar_2021}
Mohammad~Khairul Bashar, Antik Mallick, and Nikhil Shukla.
\newblock Experimental {Investigation} of the {Dynamics} of {Coupled} {Oscillators} as {Ising} {Machines}.
\newblock {\em IEEE Access}, 9:148184--148190, 2021.

\end{thebibliography}
\newpage
\appendix

\section{LagONN cost function} \label{appendix_cost_function}
Here, we describe how LagONN's cost is monitored during the simulation. From the SATlib .cnf files \cite{SATlib}, we build a system of differential equations for each instance (Eq. \ref{LagONN_dynamics}). As there is no straightforward Lyapunov function for the system, we monitor the number of unsatisfied clauses in real-time using a custom cost function $\kappa (\phi)$ defined as:
\begin{equation}
\begin{aligned}
    f_B&=C_1\bigwedge C_2 \bigwedge...\bigwedge C_{M-1}\bigwedge C_{M}\nonumber\\ \longrightarrow \kappa(\phi)&=K_1+K_2+...+K_{M-1}+K_M
    \end{aligned} \label{chap3_cost_function}
\end{equation}
where $K_m(\phi)=l_1^m\,l_2^m\,l_3^m$ with literals $l_j^m=0.5\big(1\pm\tanh(\beta\cos\phi_j^m)\big)\in[0,1]$. The sign that weights the tanh term depends on whether the literal $l_j^m$ corresponds to a positive $x_j$ (-) or negated variable $\overline{x_j}$ (+). This way, $K_m(\phi)=0$ if there is a variable assignment that satisfies the clause $C_m$. Consequently, $f_B$ is true if  $\kappa(\phi)=0$. The tanh function is used to map phases to Ising spin values and is equivalent to rounding phases to the nearest multiples of $\pi$. For a sufficiently high $\beta$-value, we then have $\kappa(\phi)=N_{unsat}$, i.e. $\kappa(\phi)$ counts the number of unsatisfied clauses. Note that with the proposed rounding procedure, a clause $C_m$ is true if and only if $K_m(\phi)<0.5^3=0.125$ ($K_m(\phi)=0.125$ when $\forall j\; \cos\phi_j^m=0)$. Thus, if $\kappa(\phi)<0.125$, $f_B$ is true and we use this value as a threshold to stop the LagONN search as shown in Fig.\ref{cost_example}a. for $N=100$ variables and $M=430$ clauses.

Fig.\ref{cost_example}b and c show the energy values $Z_m$ for each clause at the initialization (purple cross) and at the snapshot time t=150 oscillations (black star) when LagONN finds an optimal solution. For the nominal ONN, the final fixed point does not satisfy $Z_m=0$ for many clauses. For LagONN, since we stop the simulation before convergence, most of the final $Z_m$ values are zero yet.
In practice, one could have a standard Boolean circuit corresponding to the formula $f_B$ (with AND and OR gates) checking in real-time the number of satisfied clauses and sampling the phases when the cost reaches the target value $\kappa(\phi)<0.125$.

\begin{figure}[t!]
\centering
\includegraphics[width=1\textwidth]{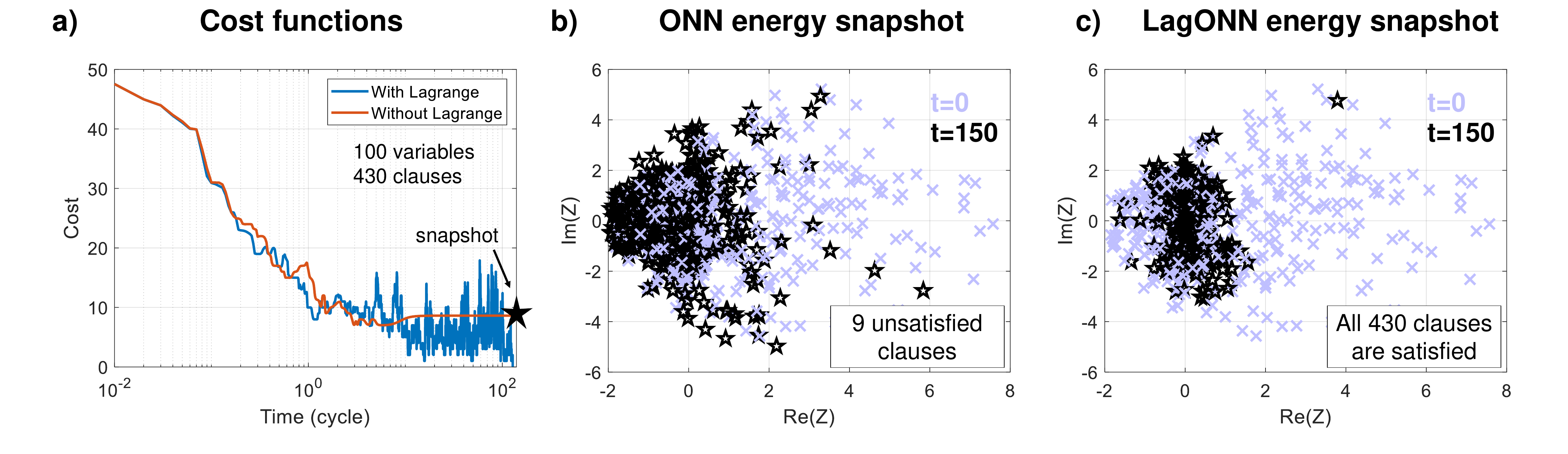}
\caption{Simulation for the satisfiable SATlib instance 'cnf-100-01' with $N=100$ variables and $M=430$ clauses. Here, we monitor the cost function and stop the simulation when the cost<0.125 or when the system reaches a fixed point. a) Cost function comparison between the standard ONN and the Lagrange version. While the two systems produce a rapid cost decrease in about 1 oscillation cycle (with more than 35 satisfied clauses), Lagrange oscillators are then actively exploring the phase space when the standard ONN gets stuck into a local minimum. When the Lagrange ONN finds an optimal phase assignment at t=150 oscillation cycles, we take a snapshot of all energy terms $Z_m$. b) Standard ONN energy snapshot. The black stars show the energy values at the snapshot time. Most of them are not settling to the target $Z_m=0$. c) LagONN energy snapshot. When the simulation is stopped (cost<0.125), $Z_m$ values are getting closer to the origin. By monitoring the cost in real-time, we do not need to wait for full convergence towards an optimal saddle point where all $Z_m=0$.}
\label{cost_example}
\end{figure}

\section{ODE solver for LagONN's state equations} \label{appendix_solver}
\begin{figure}[t!]
\centering
\includegraphics[width=0.8\textwidth]{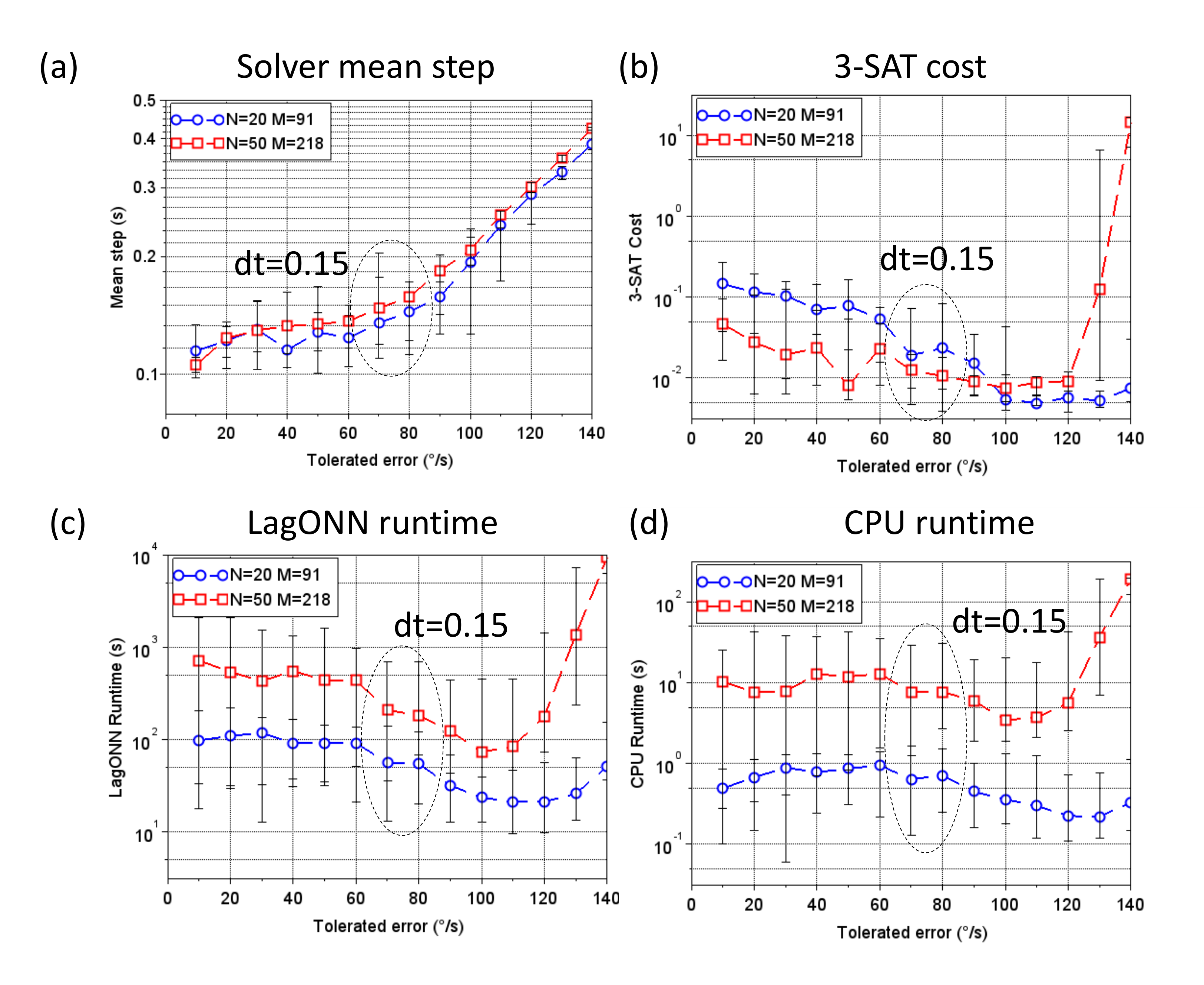}
\caption{a) Median solver step versus the tolerated phase error $\epsilon$ for 10 3-SAT instances with $(N,M)\in{(20,91),(50,218)}$. Error bars correspond to 1st and 3rd quartiles. b) 3-SAT cost vs. tolerated error. c) LagONN runtime vs. tolerated error. d) CPU runtime vs. tolerated error.}
\label{error_study}
\end{figure}
LagONN is challenging to simulate at a large scale due to the number of coupled differential equations scaling with the number of clauses. To best select a suitable ODE solver, we studied the stiffness of LagONN's state equations via a custom ODE solver in Matlab, which has an adaptive time step. We chose Fehlberg's method, which has a local error scaling as $O(dt^4)$, requiring three evaluations of $\nabla L_T(\phi)$ per time step, and is a good compromise between speed and accuracy. It consists of a predictor/corrector method that provides a local error estimate at each time step, which is then used to adapt the step size. The predictor phase values $\phi_p$ are first computed according to Heun's method as:
\begin{equation}
    \phi_p[k+1]=\phi[k]+dt[k]\times (f_1+f_2)/2
\end{equation}
where $f_1=\nabla L_T(\phi[k])$ and $f_2=\nabla L_T(\phi[k]+dt[k]\times f_1)$.

The corrector phase values are then calculated according to Simpson's rule as:
\begin{equation}
    \phi[k+1]=\phi[k]+dt[k]\times (f_1+f_2+4f_3)/6
\end{equation}
where $f_3=\nabla L_T(\phi[k]+dt[k]\times (f_1+f_2)/4)$.
We estimate the local error as:
\begin{equation}
    e_r=\sqrt{\frac{(\phi_p-\phi)\times (\phi_p-\phi)^T}{N+M}}
\end{equation}
where $N$ and $M$ are the number of variables and clauses. Given a target error $\epsilon$ in radians per second, the new time step is calculated as:
\begin{equation}
    dt[k+1]=0.9\,dt[k]\,\Gamma
\end{equation}
with $\Gamma=\sqrt{dt[k]\epsilon/e_r}$. If $\Gamma<1$, the solver reiterates with a smaller time step. Otherwise, it integrates the next point with a larger time step. We varied the tolerated error $\epsilon\leq 140^\circ$ using 10 instances per size $N=20$ and $N=50$ from the SATlib library \cite{SATlib}. Each LagONN instance was run 100 times with random initialization. Fig.\ref{error_study}a shows the solver mean step size that exponentially increases with the tolerated error $\epsilon$.

We found that LagONN is robust to numerical errors as it still finds optimal solutions with a similar runtime up to $\epsilon=100^\circ$ (Fig.\ref{error_study}b and c). LagONN's runtime even decreases with numerical error, down to $10\times$ for $N=50$. However, when $\epsilon>100^\circ$, i.e. $<dt>\approx 0.2$, LagONN's runtime significantly increases for $N=50$. Based on this study, we ran our custom solver with a fixed time step $dt=0.15$ for all simulations in the paper. The Matlab code was executed on a Linux server using one CPU per run.

\section{Algorithms} \label{appendix_SA}

\subsection{LagONN}
LagONN's pseudo code for Max-3-SAT is shown next, where we first construct the Lagrange function $L$ based on the input 3-SAT instance, and express its analytical derivatives used later on by the ODE solver for integration. In practice, both gradients of $L$ in the $\phi$ and $\phi_\lambda$ directions are computed by dedicated functions called by the solver "integrate". We used the custom ODE solver previously described in Appendix \ref{appendix_solver}. We also construct a cost function $\kappa$ for the instance, which is called at each iteration to stop the dynamics when the phase assignment is optimal, as described in Appendix \ref{appendix_cost_function}.
\begin{algorithm}[H]
\caption{Simulate LagONN for satisfiable 3-SAT}
\KwIn{3-SAT instance, Number of trials $\texttt{MAX\_TRIALS}$, Simulation time \texttt{MAX\_TIME}, solver time step $dt$}
\KwOut{A phase assignement $\phi$}
Create Lagrange function $L$ for the instance\;
Create gradient functions $\nabla_\phi L$ and $\nabla_{\phi_\lambda}L$\;
Create cost function for the instance $\kappa$\;
$N_{\text{iteration}} \gets \lfloor\texttt{MAX\_TIME}/dt\rfloor$ \;
\For{$\text{trial} \gets 1$ \KwTo $\texttt{MAX\_TRIALS}$}{
    $(\phi,\phi_\lambda) \gets$ random phase\;
    \For{$i \gets 1$ \KwTo $N_{\text{iteration}}$}{
        $(\phi,\phi_\lambda) \gets \text{integrate}(\phi,\phi_\lambda,\nabla_\phi L,\nabla_{\phi_\lambda}L,dt)$\;
       \If{$\kappa (\phi)<0.125$}{
         \Return{$\phi$};
       }
    }
}
\end{algorithm}

\subsection{Simulated Annealing}
In this paper, we execute SASAT, the simulated annealing algorithm proposed in \cite{johnson_1996} for SAT and expressed as follows:

\begin{algorithm}[H]
\caption{SASAT Algorithm \cite{johnson_1996}}
\KwIn{A set of clauses with $N$ variables, \texttt{MAX\_TRIALS}, \texttt{MAX\_TEMP}, and \texttt{MIN\_TEMP}}
\KwOut{A variable assignement $S$}

$trial \gets 1$\;
$decay\_rate\gets 0.2/N$\;
\While{$trial \leq \texttt{MAX\_TRIALS}$}{
    $S \gets$ a random variable assignment\;
    $j \gets 0$\;
    $T \gets \texttt{MAX\_TEMP}$\;
    \While{$T \geq $\texttt{MIN\_TEMP}}{
        \If{$S$ satisfies the clauses}{
            \Return{$S$}\;
        }
        $T \gets \texttt{MAX\_TEMP} \cdot \exp{(-j \cdot decay\_rate)}$\;
        \For{$i \gets 1$ to $N$}{
            Compute the cost change $\delta$ in the number of unsatisfied clauses if $i$ is flipped\;
            Flip $i$ with probability $1/(1 + \exp{(\frac{\delta}{T})})$\;
            $S \gets$ the new assignment if flipped\;
        }
        $j \gets j + 1$\;
        }
    $trial \gets trial + 1$\;
}
\end{algorithm}
The temperature decay rate was adjusted heuristically and scaled with the number of variables as $decay\_rate=0.2/N$ to increase the annealing time and the success probability for larger instances. We have set \texttt{MAX\_TEMP}=1 and \texttt{MIN\_TEMP}=0.01 for all the experiments by inspecting the cost trace for several sizes.

\section{LagONN saddle points and dynamics}\label{appendix_dynamics}
\subsection{Optimal saddle point}
Here we prove Theorem \ref{theorem2} which generalizes Theorem \ref{theorem1} to $M$ clauses.
\originaltheorem*
\begin{proof}
\begin{enumerate}
    \item We motivate the search for a saddle point using the concept of \textit{duality} \cite{boyd_2023}. Consider the following dual function defined as:
\begin{equation}
    D_T(\phi_\lambda)=\min_\phi L_T(\phi,\phi_\lambda)
\end{equation}
For any vector of phases $\phi_\lambda$ it is possible to find an optimal assignment of phase $\phi^*$ such that constraints are satisfied, i.e. for all clauses $Z_m=0$ which gives as optimal value $L_T(\phi^*,\phi_\lambda)=0$. Hence, $D_T(\phi_\lambda)\leq 0$.
The dual problem consists of finding the best lower bound for the optimal value of the initial problem — that is satisfying the constraints $Z_m=0$ for all clauses $C_m$. Hence, we are looking for $\phi_\lambda$ that maximizes $D_T(\phi_\lambda)$ as:
\begin{equation}
    \max_{\phi_\lambda} D_T(\phi_\lambda)=\max_{\phi_\lambda}\min_\phi L_T(\phi,\phi_\lambda)
\end{equation}
For any $\phi$, we can find a vector of phases $\phi_\lambda$ such that their corresponding unitary vectors are orthogonal to their corresponding $\Vec{Z_m}$, hence, $\max_{\phi_\lambda}\min_\phi L_T(\phi,\phi_\lambda)=0$.
Consider now the inverse situation where we first maximize $L_T$ as:
\begin{equation}
    P_T(\phi)=\max_{\phi_\lambda} L_T(\phi,\phi_\lambda)
\end{equation}
For any $\phi$, we can set the phases $\phi_\lambda$ such that their corresponding unitary vectors point in the same direction as their $\Vec{Z_m}$. Hence, $P_T(\phi)\geq 0$. Seeking the best higher bound is expressed as:
\begin{equation}
    \min_{\phi} P_T(\phi)=\min_{\phi}\max_{\phi_\lambda}L_T(\phi,\phi_\lambda)
\end{equation}
For any $\phi_\lambda$, the best higher bound $P_T(\phi)$ is obtained when all the constraints are satisfied, i.e. $Z_m=0$ for all clauses $C_m$. Hence, $\min_{\phi}\max_{\phi_\lambda}L_T(\phi,\phi_\lambda)=0$.
In summary, we obtain:
\begin{equation}
    L(\phi^*,\phi_\lambda^*)=\max_{\phi_\lambda}\min_\phi L_T(\phi,\phi_\lambda)=\min_{\phi}\max_{\phi_\lambda}L_T(\phi,\phi_\lambda)=0
    \label{strong_duality}
\end{equation}
Since $L_T$ is continuous in both $\phi$ and $\phi_\lambda$, and $L(\phi^*,\phi_\lambda^*)$ is attained for the satisfiable assignment of phase $\phi^*\in\{0;\pi\}^N$, $L(\phi^*,\phi_\lambda^*)$ is a saddle point satisfying the inequality $L_T(\phi^*,\phi_\lambda)\leq L_T(\phi^*,\phi_\lambda^*) \leq L_T(\phi,\phi_\lambda^*) $.
\item By contradiction, suppose that $(\phi^*,\phi_\lambda^*)$ is a saddle point of $L_T$ satisfying $L_T(\phi^*,\phi_\lambda)\leq L_T(\phi^*,\phi_\lambda^*) \leq L_T(\phi,\phi_\lambda^*) $, but there is some clause $m$ such that $Z_m\neq 0$. Let us consider the corresponding Lagrange oscillator $\Vec{u_\lambda^m}$: the term $\Vec{u_\lambda^m}.\Vec{Z_m}$ is maximum when $\Vec{u_\lambda^m}$ points in the same direction as $\Vec{Z_m}$. If at the saddle point $\Vec{u_\lambda^m}$ is not already pointing in the same direction as $\Vec{Z_m}$, we can find a new Lagrange phase $\hat{\phi_\lambda}$ for this clause such that $L_T(\phi^*,\hat{\phi_\lambda})>L_T(\phi^*,\phi_\lambda^*)$, which violates the saddle condition $L_T(\phi^*,\phi_\lambda)\leq L_T(\phi^*,\phi_\lambda^*)$ for all $\phi_\lambda$. If at the saddle point, the vector $\Vec{u_\lambda^m}$ with phase $\phi_\lambda^*$ points in the same direction as its respective $\Vec{Z_m}\neq 0$, we can find a new assignment of phase $\hat{\phi}$ satisfying all constraints and $Z_m=0$ such that $L_T(\phi^*,\phi_\lambda^*)>L_T(\hat{\phi},\phi_\lambda^*)$, violating the saddle condition $L_T(\phi^*,\phi_\lambda^*)\leq L_T(\phi,\phi_\lambda^*)$ for any $\phi$.
\end{enumerate}

\end{proof}
\subsection{Proposed dynamics to find a saddle point}
To find an optimal saddle point as described by Theorem \ref{theorem2}, we combine gradient descent and ascent along $\phi$ and $\phi_\lambda$:
\begin{equation}
    \begin{cases}
        \tau \dot{\phi_x}=-\nabla_{\phi_x} L_T\\
        \tau_\lambda\dot{\phi_\lambda}=+\nabla_{\phi_\lambda} L_T
    \end{cases}
\end{equation}
Under these dynamics, the Lagrange function is not a Lyapunov function for the system since with our proposed dynamics $L_T$ can increase with time (gradient ascent along $\phi_\lambda$). Focusing on a single clause $i$, its time derivative is indeed expressed as:
\begin{equation}
    \frac{dL_i}{dt}=\frac{d\Vec{u_\lambda}}{dt}.\Vec{Z_i}+\frac{d\Vec{Z_i}}{dt}.\Vec{u_\lambda}
\end{equation}
and $L_i$'s gradient descent causes $\Vec{Z_i}$ to evolve in the opposite direction from $\Vec{u_\lambda}$ as expressed here for $\tau=1$:
\begin{equation}
\begin{aligned}
    \frac{d\Vec{Z_i}}{dt}.\Vec{u_\lambda}&=\frac{\partial\Vec{Z_i}}{\partial \phi_X}.\Vec{u_\lambda}\frac{d\phi_X}{dt}+\frac{\partial\Vec{Z_i}}{\partial \phi_Y}.\Vec{u_\lambda}\frac{d\phi_Y}{dt}+\frac{\partial\Vec{Z_i}}{\partial \phi_Z}.\Vec{u_\lambda}\frac{d\phi_Z}{dt} \\ 
    &= -\big( \frac{\partial\Vec{Z_i}}{\partial \phi_X}.\Vec{u_\lambda}\big)^2-\big( \frac{\partial\Vec{Z_i}}{\partial \phi_Y}.\Vec{u_\lambda}\big)^2-\big( \frac{\partial\Vec{Z_i}}{\partial \phi_Z}.\Vec{u_\lambda}\big)^2 \\
    &\leq 0
    \end{aligned}
\end{equation}
whereas $L_i$'s gradient ascent tends to bring $\Vec{u_\lambda}$ towards $\Vec{Z_i}$ as:
\begin{equation}
\begin{aligned}
    \frac{d\Vec{u_\lambda}}{dt}.\Vec{Z_i}&=\frac{\partial\Vec{u_\lambda}}{\partial\phi_\lambda}\frac{d\phi_\lambda}{dt}.\Vec{Z_i}\\
    &=(\Vec{Z_i}.\Vec{u_\lambda'})^2 \\
    &\geq 0
    \end{aligned}
\end{equation}

\section{LagONN stability for unsatisfiable problems} \label{appendix_stability}
\begin{figure}[t!]
\centering
\includegraphics[width=0.7\textwidth]{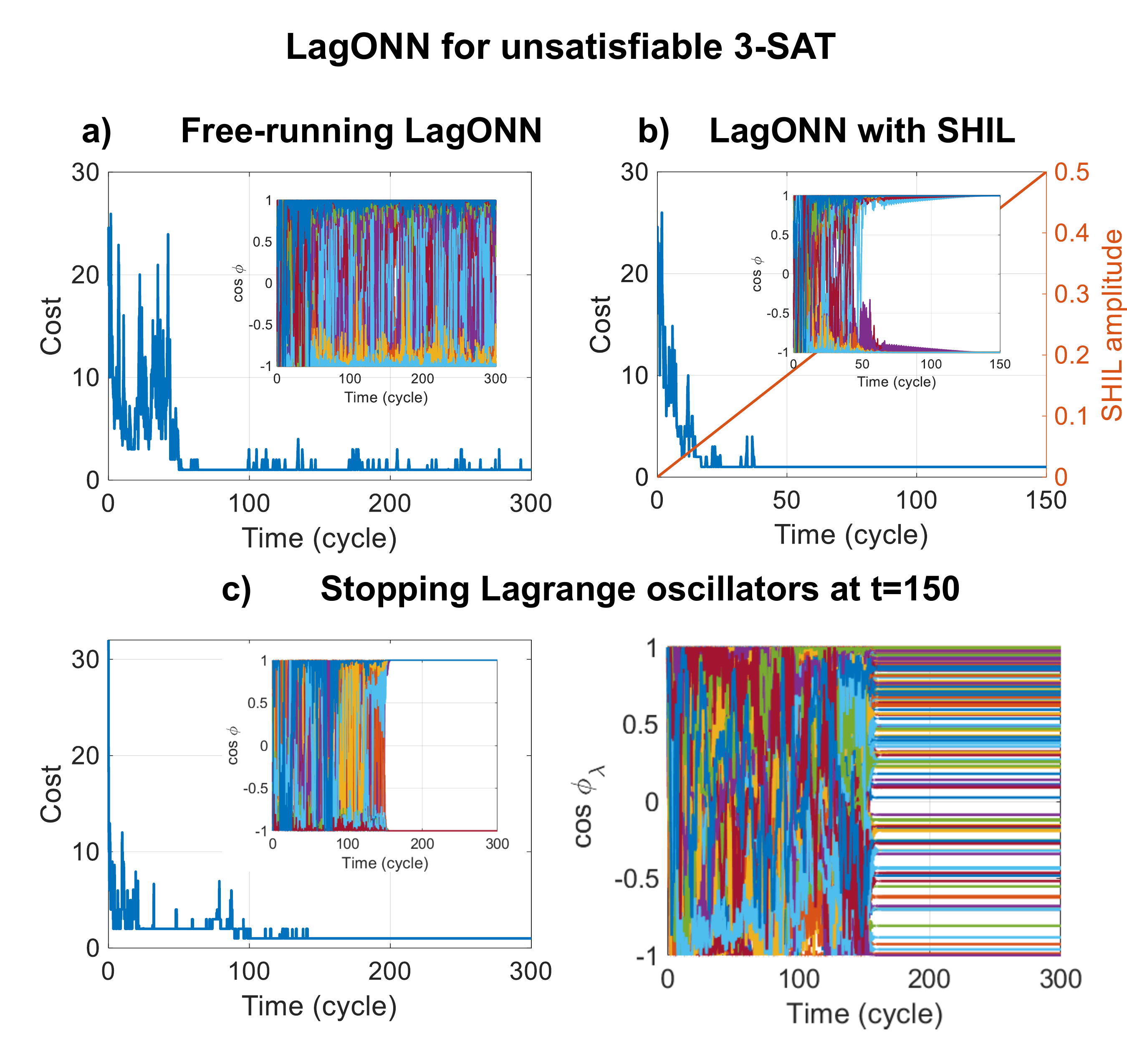}
\caption{LagONN stability study for unsatisfiable 3-SAT instances (50 variables and 218 clauses from SATLib). The insets show the phase evolution for the $N$ SAT variables. a) LagONN running in its nominal mode with all oscillators activated. b) LagONN with second harmonic injection locking (SHIL) to enforce phase binarization for the SAT variables. The SHIL amplitude is slowly ramped up until stabilization, similar to an annealing schedule. c) LagONN simulation with the Lagrange oscillators stopped at t=150 oscillations. The right plot shows the Lagrange phases $\phi_\lambda$ remaining constant after the stop.}
\label{lagoon_stability}
\end{figure}

The stability of LagONN's dynamics is not guaranteed due to the gradient ascent mechanism introduced by the Lagrange oscillators. If the constrained problem is not satisfiable, such as an unsatisfiable 3-SAT formula, the oscillators never settle, as shown in Fig.\ref{lagoon_stability}a for a 50-variable unsatisfiable instance from SATlib. One way of stabilizing the system is to stop the Lagrange oscillator evolution for all clauses $m$ as $d\phi_\lambda^m/dt=0$, for instance by setting the synaptic amplitude to 0 corresponding to a Lagrange time constant $\tau_\lambda\rightarrow +\infty$. In that case, the network has a Lyapunov function $L$ which is the Lagrange function itself. $L$ decreases over time as
\begin{equation}
\begin{aligned}
    \frac{d}{dt} L(\phi, \phi_\lambda) 
=& \sum_j \frac{\partial L}{\partial \phi_j} \frac{d\phi_j}{dt} +\sum_m \frac{\partial L}{\partial \phi_\lambda^m} \frac{d\phi_\lambda^m}{dt}  \\
=& -\tau \sum_j \left( \frac{d\phi_j}{dt} \right)^2 
\le 0
\end{aligned}
\end{equation}
which induces global stability as $L$ is bounded from below.
An example of dynamics when stopping the Lagrange evolution at t=150 oscillations is shown in Fig.\ref{lagoon_stability}c, where the system settles to the optimal Max-3-SAT solution (one remaining unsatisfied clause). Stopping the Lagrange oscillators after some time is different from starting the dynamics without Lagrange oscillators, as the Lagrange phases $\phi_\lambda$ evolve before getting frozen to their final value.

Another option to enforce stability is the use of second harmonic injection locking (SHIL) to phase-lock oscillators to binary values. Since this technique is used in practice to mitigate oscillator frequency variations and recover binary Ising spins \cite{Bashar_2021,Wang_2021}, slowly annealing the injection strength could further freeze phases in the end before read-out, as shown with the simulation in Fig.\ref{lagoon_stability}b where we added a potential function $V_{SHIL}=-K(t)\sum_i^N\cos(2\phi_i)$ to the Lagrange function. In practice, $K(t)\geq 0$ is the increasing amplitude of a harmonic signal injected at $2\omega_0$, with $\omega_0$ the mean oscillator frequency. This injection could help scale up the system to overcome variability and stability issues in physical implementations.

\section{Impact of Lagrange oscillator speed} \label{appendix_lagrange_speed}
Throughout the paper, we assumed that the Lagrange and the standard oscillators are equally fast, i.e. $\tau=\tau_\lambda=1$ in Eq. \ref{LagONN_dynamics}. Here, we study how speeding up or slowing down the Lagrange oscillators affects the runtime of the whole network. In particular, we set $\tau=1$ and vary $\tau_\lambda$ for the Lagrange oscillator and measure the resulting time-to-solution (TTS) (see Eq. \ref{eq:timetosolution}). 

Here we use the first $20$ instances from SATlib with $20$ and $50$ variables with $100$ trials for each instance. For each $\tau_\lambda$, we set a maximum simulation time $t_{max}$. Next, we check whether each trial finds an optimal solution in the predetermined simulation time. Based on these results, we compute the median success probability $\text{p}_\text{s}$ for each $\tau_\lambda$. Combining the maximum simulation time and the success probability, we can quantify the median time to solution for $20$ and $50$ variables for each $\tau_\lambda$ value. The results are summarized in Fig. \ref{tau_runtimes}.
\begin{figure}[t!]
\centering
\includegraphics[width=0.8\textwidth]{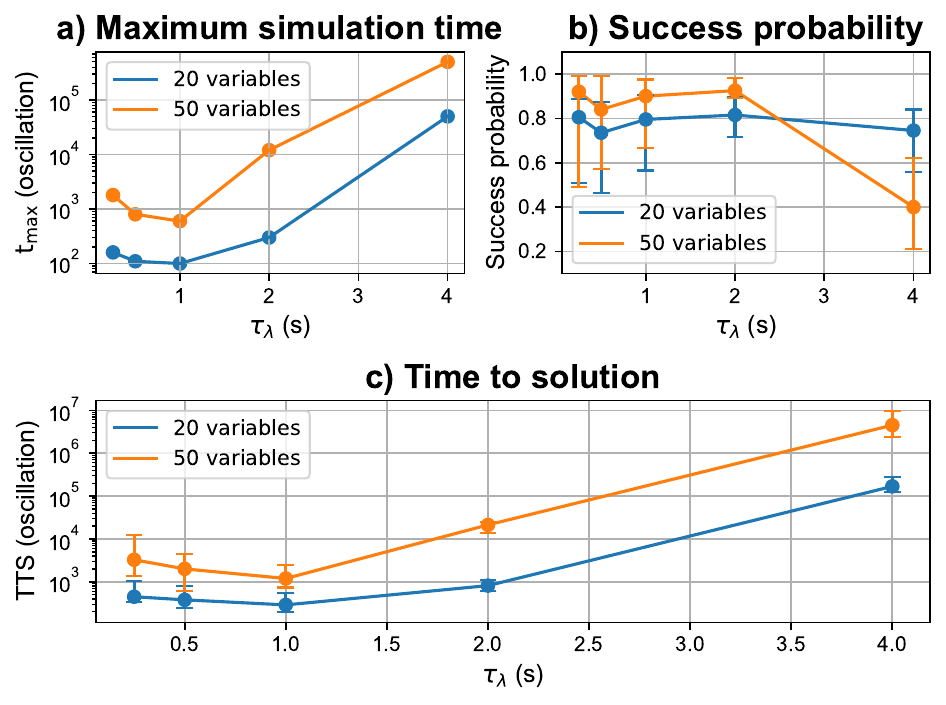}
\caption{Impact of different $\tau_\lambda$ on the time-to-solution of LagONN. Here a simple benchmark on the Max-3-SAT problem with $20$ and $50$ variables for $20$ instances each and $100$ trials per instance. a.) Maximum simulation time as a function of $\tau_\lambda$ for $20$ and $50$ variables. b.) Median success probability $\text{p}_\text{s}$ as a function of $\tau_\lambda$ for $20$ and $50$ variables. c.) Median time to solution as a function of $\tau_\lambda$ for $20$ and $50$ variables.}
\label{tau_runtimes}
\end{figure}

For minimum TTS, there seems to be an optimal value for $\tau_\lambda\approx 1$, which means all oscillators should have the same speed. Although more rigorous analysis is needed, $\tau_\lambda/\tau=1$ may be linked to the saddle geometry of the optimal point (Theorem \ref{theorem2}), where minimization and maximization are interchangeable and maximization does not need to be performed in an outer loop as a slower process.

\section{Discretization of LagONN dynamics} \label{appendix_discretization}

To lay the groundwork for a potential LagONN digital hardware implementation, here we study the effect of phase discretization on the dynamics.
We discretize the phase interval from $0$ to $2\pi$ with a fixed number of states $N_\text{states}$.
For example for $N_\text{states} = 16$ the phases can take on values
\begin{equation}
    \phi = \frac{2\pi}{16}k,
\end{equation}
with $k$ an integer from 0 to $N_\text{states} - 1 = 15$.
We study $N_\text{states} = 16, 32, 64, 128, 256, 512, 1024, 2048, 4096, 8192$ and run LagONN simulations on the first 20 instances from SATLib for 20 variables and 91 clauses, and for 50 variables and 218 clauses.
Each instance is given 100 trials with random initial phases.
We used the same random initial phases per trial, except discretized to the corresponding $N_\text{states}$.
The time-to-solution metric, as defined in Eq. \ref{eq:timetosolution}, is used to compare different discretization levels.
The results are shown in Fig. \ref{fig:discretization_tts}.
For the case of 50 variables with $N_\text{states} = 16$ or $N_\text{states} = 32$, the data points are excluded due to unstable behavior.

\begin{figure}[t!]
    \centering
    \includegraphics[width=0.7\textwidth]{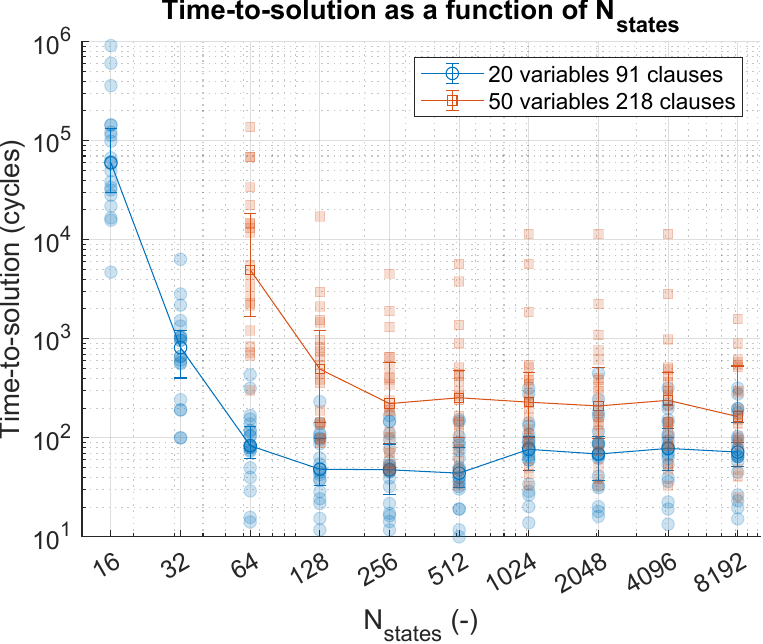}
    \caption{Performance scaling when discretizing the phase dynamics of LagONN to different number of states $N_\text{states}$. The median runtimes are shown with the $25\%$ and $75\%$ quantiles as error bars.}
    \label{fig:discretization_tts}
\end{figure}

As can be seen in Fig. \ref{fig:discretization_tts}, there is a strong increase in time to solution below $N_\text{states} = 64$ for instances with 20 variables, after which the TTS flattens.
For 50 variables, the same can be seen, except that the trend flattens after around $N_\text{states} = 128$.
We can conclude from these results that for a digital implementation of LagONN, one should take care to discretize the phases to a sufficiently high number of states to obtain a stable system.
For 20 variables, this will be around $N_\text{states} \geq 64$, while for 50 variables it is around $N_\text{states} \geq 128$.
Although more data is needed, we hypothesize that a higher number of variables beyond 50 would also require a higher number of $N_\text{states}$.

\end{document}